\documentclass[11pt]{article}
\usepackage{odonnell}

\newcommand{\Opt}{\mathrm{Opt}}

\newcommand{\pE}{\mathop{\bf \widetilde{E}\/}}

\newcommand{\alphabet}{\Omega}
\newcommand{\alphasize}{q}
\newcommand{\x}{x}
\newcommand{\indet}[2]{1_{#1}(\x_{#2})}
\newcommand{\olindet}[2]{\ol{1_{#1}}(\x_{#2})}

\newcommand{\olbindet}[2]{\ol{1_{#1}}(\bx_{#2})}

\newcommand{\degmlin}{\deg_{\mathrm{mlin}}}
\newcommand{\multilineardegree}{multilinear-degree\xspace}

\newcommand{\cl}{\mathrm{cl}}

\newcommand{\vbls}{\mathrm{vbls}}
\newcommand{\cons}{\mathrm{cons}}
\newcommand{\edges}{\mathrm{edges}}
\newcommand{\Size}[1]{|\cons(#1)|}

\newcommand{\vblsp}{\mathrm{vbls}}

\newcommand{\CONSMALL}{\scalebox{.75}[1.0]{\textnormal{SMALL}}}
\newcommand{\maxarity}{K}
\newcommand{\Tcxty}{\tau}
\newcommand{\Tone}{\lambda}

\newcommand{\price}{\zeta}
\newcommand{\density}{\Delta}

\newcommand{\subfactor}{$\Tcxty$\mbox{-}\nolinebreak\hspace{0pt}subgraph\xspace}
\newcommand{\subfactors}{$\Tcxty$\mbox{-}\nolinebreak\hspace{0pt}subgraphs\xspace}

\newcommand{\subfactorplus}{$\Tcxty$\mbox{-}\nolinebreak\hspace{0pt}subgraph${}^+$\xspace}
\newcommand{\subfactorpluses}{$\Tcxty$\mbox{-}\nolinebreak\hspace{0pt}subgraphs${}^+$\xspace}
\newcommand{\smallish}{small\xspace}
\newcommand{\smallness}{smallness\xspace}

\newcommand{\prev}{\mathrm{pr}}

\newcommand{\PvZ}{\scalebox{.75}[1.0]{\textnormal{PvZ}}}
\newcommand{\PvZsubscript}{\scalebox{.5625}[.75]{\textnormal{PvZ}}}

\newcommand{\cmplx}{\calC}
\newcommand{\CSP}{\mathrm{CSP}}
\newcommand{\Lovasz}{Lov\'{a}sz\xspace}
\newcommand{\SDPOpt}{\mathrm{SDPOpt}}

\usepackage{enumitem}

\makeatletter
\newcommand{\manuallabel}[2]{\def\@currentlabel{#2}\label{#1}}
\makeatother

\begin{document}

\title{Sum of squares lower bounds for refuting any CSP}

\author{Pravesh K. Kothari\thanks{Princeton University and IAS. \url{kothari@cs.princeton.edu}} \and Ryuhei Mori\thanks{Department of Mathematical and Computing Sciences, Tokyo Institute of Technology. \url{mori@is.titech.ac.jp}} \and  Ryan O'Donnell\thanks{Computer Science Department, Carnegie Mellon University. Supported by NSF grant CCF-1618679. \{odonnell,dwitmer\}@cs.cmu.edu} \and David Witmer${}^\ddag$}

\maketitle

\begin{abstract}
Let $P:\{0,1\}^k \to \{0,1\}$ be a nontrivial $k$-ary predicate. Consider a random instance of the constraint satisfaction problem $\CSP(P)$ on $n$ variables with $\Delta n$ constraints, each being $P$ applied to $k$ randomly chosen literals.  Provided the constraint density satisfies $\Delta \gg 1$, such an instance is unsatisfiable with high probability. The \emph{refutation} problem is to efficiently find a proof of unsatisfiability.


We show that whenever the predicate $P$ supports a $t$-\emph{wise uniform} probability distribution on its satisfying assignments, the sum of squares (SOS) algorithm of degree \mbox{$d = \Theta(\frac{n}{\Delta^{2/(t-1)} \log \Delta})$} (which runs in time $n^{O(d)}$) \emph{cannot} refute a random instance of $\CSP(P)$. In particular, the polynomial-time SOS algorithm requires $\wt{\Omega}(n^{(t+1)/2})$ constraints to refute random instances of CSP$(P)$ when $P$ supports a $t$-wise uniform distribution on its satisfying assignments. 
Together with recent work of Lee~et~al.~\cite{LRS15}, our result also implies that \emph{any} polynomial-size semidefinite programming relaxation for refutation requires at least $\wt{\Omega}(n^{(t+1)/2})$ constraints.

More generally, we consider the $\delta$-refutation problem, in which the goal is to certify that at most a $(1-\delta)$-fraction of constraints can be simultaneously satisfied.  We show that if $P$ is $\delta$-close to supporting a $t$-wise uniform distribution on satisfying assignments, then the degree-$\Theta(\frac{n}{\Delta^{2/(t-1)} \log \Delta})$ SOS algorithm cannot $(\delta+o(1))$-refute a random instance of CSP$(P)$.   
This is the first result to show a distinction between the degree SOS needs to solve the refutation problem and the degree it needs to solve the harder $\delta$-refutation problem.

Our results (which also extend with no change to CSPs over larger alphabets) subsume all previously known lower bounds for semialgebraic refutation of random CSPs.  For every constraint predicate~$P$, they give a three-way hardness tradeoff between the density of constraints, the SOS degree (hence running time), and the strength of the refutation. By recent algorithmic results of Allen~et~al.~\cite{AOW15} and Raghavendra~et~al.~\cite{RRS16}, this full three-way tradeoff is \emph{tight}, up to lower-order factors.

\end{abstract}

\setcounter{page}{0}
\thispagestyle{empty}
\newpage

\section{Introduction}  \label{sec:intro}

Where are the hard problems?

In computational complexity, we have a comprehensive theory of worst-case hardness, assuming $\PTIME \neq \NP$. The theory is particular rich in the context of constraint satisfaction problems (CSPs) --- optimization tasks that are both simple to state and powerfully expressive. (See, e.g., \cite{BJK05, Rag08}.) But despite our many successes in the theory of $\NP$-completeness and $\NP$-hardness-of-approximation, we know relatively little about the nature of hard instances.    For example, $3$-SAT is conjecturally hard to solve --- or even approximate to factor $\frac78 + \eps$ --- in $2^{o(n)}$ time.  But what do hard(-seeming) instances look like?  How can we generate one?  These sorts of questions are a key part of understanding what makes  various algorithmic problems truly hard.  They are particularly important for CSPs, as these are nearly always the starting point for hardness reductions; the ability to find hard instances for CSPs yields the ability to find hard instances for many other algorithmic problems.

In some sense, a single instance can never be ``hard'' because its solution can always be hard-coded into an algorithm.  Thus it is natural to turn to \emph{random} instances, and the theory of average-case hardness.
Uniformly random instances of CSPs are a particularly simple and natural  source of hard(-seeming) instances.  Furthermore, they arise as the fundamental object of study in many disparate areas of research, including cryptography \cite{ABW10}, proof complexity \cite{BB02}, hardness of approximation \cite{Fei02}, learning theory \cite{DLS14}, SAT-solving \cite{SAT14}, statistical physics \cite{CLP02}, and combinatorics.
\subsection{Random CSPs} \label{sec:csps}
Let $\Omega$ be a finite alphabet and let $\calP$ be a collection of nontrivial predicates $\Omega^k \to \{0,1\}$.  An input~$\calI$ to the problem $\CSP(\calP)$ consists of $n$ variables $x_1, \dots, x_n$, along with a list $\calE$ of~$m$ constraints $(P, S)$, where $P$ is a predicate from $\calP$, and $S \in [n]^k$ is a scope of $k$ distinct variables.  We often think of the associated ``factor graph'': that is, the bipartite graph with $n$~``variable-vertices'', $m$~``constraint-vertices'' of degree~$k$, and edges defined by the scopes.

Given~$\calI$, the algorithmic task is to find an assignment to the variables so as to maximize the fraction of satisfied constraints, $\avg_{(P,S) \in \calE} P(x_{S_1}, \dots, x_{S_k})$.    We write $\Opt(\calI)$ for the maximum possible fraction, and say that $\calI$ is \emph{satisfiable} if $\Opt(\calI) = 1$.  For a fixed \emph{constraint density} $\Delta = \Delta(n) > 0$, a \emph{random} instance of $\CSP(\calP)$ is defined simply by choosing $m = \Delta n$ constraints uniformly at random: random scopes and random $P \in \calP$.

The most typical examples involve a binary alphabet $\Omega = \{0,1\}$, a fixed predicate $P : \{0,1\}^k \to \{0,1\}$, and $\calP = P^{\pm}$, where by $P^{\pm}$ we mean the collection of all~$2^k$ predicates obtained by letting~$P$ act on possibly-negated input bits (``literals'').  For example, if $P$ is the $k$-bit logical OR function, then $\CSP(P^\pm)$ is simply the $k$-SAT problem.  In this introductory section, we'll focus mainly on these kinds of CSPs.

For random CSPs, the constraint density $\Delta$ plays a critical role; naturally, the larger it is, the more likely $\calI$ is to be unsatisfiable.  For a fixed $\calP$, it is easy to show the existence of constants $\alpha_0 < \alpha_1$ such that when $\Delta < \alpha_0$, a random instance $\calI$ of $\CSP(P)$ is satisfiable with high probability (whp), and when $\Delta > \alpha_1$, $\calI$ is unsatisfiable whp. For most interesting $\calP$, it is conjectured that there is even a \emph{sharp threshold} $\alpha_0 = \alpha_1 = \alpha_c$.  (This has been proven for $k$-SAT with $k$ large enough~\cite{DSS15}.  See~\cite{CD09} for a characterization of those Boolean CSPs for which a sharp threshold is expected.)

For random instances with subcritical constraint density, $\Delta < \alpha_c$, the natural algorithmic task is to try to efficiently find satisfying assignments.  There have been quite a few theoretical and practical successes for this problem, for $\Delta$ quite large and even approaching~$\alpha_c$ \cite{Gab16, MPR16}.  On the other hand, for random instances with supercritical constraint density, $\Delta > \alpha_c$, the natural algorithmic task is to try to efficiently \emph{refute} them; i.e., produce a certificate of unsatisfiability.  For many CSPs, this task seems much harder, even heuristically.  For example, random $3$-SAT instances are unsatisfiable (whp) once $\Delta > 4.49$ \cite{DKM08}; however, even for $\Delta$ as large as~$n^{.49}$ there is no known algorithm that efficiently refutes random instances --- even heuristically/experimentally.  Thus the refutation task for random instances of CSPs with many constraints may be a source of simple-to-generate, yet hard-to-solve problems.

\subsection{The importance and utility of hardness assumptions for random CSPs}
In this section, we discuss the task of refuting random CSP instances and the importance of understanding the ``constraint density vs.\ running time vs.\ refutation strength tradeoff'' for \emph{all} predicate families~$\calP$.  To define our terms, a \emph{(weak) refutation algorithm} for $\CSP(\calP)$ is an algorithm that takes as input an instance~$\calI$ and either correctly outputs ``unsatisfiable'', or else outputs ``don't know''.  For a given density~$\Delta$ (larger than the critical density), we say the algorithm ``succeeds'' if it outputs ``unsatisfiable'' with high probability (over the choice of~$\calI$, and over its internal coins, if any).  More generally, we can consider refutation algorithms that always output a correct upper bound on~$\Opt(\calI)$; we call them \emph{$\delta$-refutation algorithms} if they output an upper bound of $1-\delta$ (or smaller) with high probability.  The case of $\delta = 1/m$, where $m = \Delta n$ is the number of constraints, corresponds to the simple weak refutation task described earlier (with an output of ``$1$'' corresponding to ``don't know'').  In general, we refer to~$\delta$ as the ``strength'' of the refutation.

For a wide variety of areas --- cryptography, learning theory, and approximation algorithms --- it is of significant utility to have concrete hardness assumptions concerning random CSPs.  Because uniformly random CSPs are very simply and concretely defined, they form an excellent basis for constructing other potentially hard problems by reduction.  An early concrete hypothesis comes from an influential paper of Feige~\cite{Fei02}:

\paragraph{Feige's R3SAT Hypothesis.} \emph{For every small $\delta > 0$ and for large enough constant~$\Delta$, there is no polynomial-time algorithm that succeeds in $\delta$-refuting random instances of $3$\textnormal{-SAT}.}\\

Feige's main motivation was hardness of approximation; e.g., he showed that the R3SAT Hypothesis implies stronger hardness of approximation results than were previously known for several problems (Balanced Bipartite Clique, Min-Bisection, Dense $k$-Subgraph, $2$-Catalog).  By reducing from these problems, several more new hardness of approximation results based on Feige's Hypothesis have been shown in a variety of domains \cite{BKP04, DFHS06, Bri08, AGT12}.  Feige \cite{Fei02} also related hardness of refuting $3$-SAT to hardness of refuting $3$-XOR.  The assumption that refuting $3$-XOR is hard has been used to prove new hardness results in subsequent work \cite{OWWZ14}.  Alekhnovich \cite{Ale03} further showed that certain average-case hardness assumptions for XOR imply additional hardness results, as well as the existence of secure public key cryptosystems.

In even earlier cryptography work, Goldreich~\cite{Gol00} proposed using the average-case hardness of random CSPs as the basis for candidate one-way functions.  Subsequent work (e.g., \cite{MST03}) suggested using similar functions as candidate pseudorandom generators (PRGs).  The advantage of this kind of construction is the extreme simplicity of computing the PRG: indeed, its output bits can be computed in $\mathsf{NC}^0$, constant parallel time.  Further work investigated variations and extensions of Goldreich's suggestion~\cite{ABW10, ABR12, AL16}; see Applebaum's survey~\cite{App13b} for many more details.  Of course, the security of these candidate cryptographic constructions depends heavily on the hardness of refuting random CSPs.  Applebaum, Ishai, and Kushilevitz \cite{AIK06} took a slightly different approach to showing that PRGs exist in $\mathsf{NC}^0$, instead basing their result on one of Alekhnovich's average case XOR hardness assumptions \cite{Ale03}.

Finally, a recent exciting sequence of works due to Daniely and coauthors~\cite{DLS13, DLS14, DSS14, Dan15} has linked hardness of random CSPs to hardness of learning.  By making concrete conjectures about the hardness of refuting random $\CSP(\calP)$ for various $\calP$ and for superpolynomial~$\Delta$, they obtained negative results for several longstanding problems in learning theory, such as learning DNFs and learning halfspaces with noise.

\subsection{Desiderata for hardness results} \label{sec:desire}
While Feige's R3SAT Hypothesis has proven useful in hardness of approximation, there are several important strengthenings of it that would lead to even further utility.  We discuss here four key desiderata for hardness results about random CSPs:

\begin{enumerate}
\bfseries \item Predicates other than SAT. \mdseries  The hardness of random $3$-SAT and $3$-XOR has been most extensively studied, but for applications it is quite important to consider other predicates.  For hardness of approximation, already Feige~\cite{Fei02} noted that he could prove stronger inapproximability for the $2$-Catalog problem assuming hardness of refuting random $k$-AND for large~$k$.  Subsequent work has used assumptions about the hardness of refuting CSPs with other predicates to prove additional worst-case hardness results~\cite{GL04, AAM+11, CMVZ12, BCMV12, RSW16}. Relatedly, Barak, Kindler, and Steurer~\cite{BKS13} have recently considered a generalization of Feige's Hypothesis to all Boolean predicates, in which the assumption is that the ``basic SDP'' provides the best $\delta$-refutation algorithm when $\Delta = O(1)$.   They also describe the relevance of predicates over larger alphabet sizes and with superconstant arity for problems such as the Sliding Scale Conjecture and Densest $k$-Subgraph.  Bhaskara et al.~\cite{BCGVZ12} prove an SOS lower bound for Densest $k$-Subgraph via a reduction from Tulsiani's SOS lower bound for random instances of  CSP$(P)$ with $P$ a $q$-ary linear code \cite{Tul09}.  A computational hardness assumption for refutation of this CSP would therefore give a hardness result for Densest $k$-Subgraph.

Regarding cryptographic applications, the potential security of Goldreich's candidate PRGs depends heavily on what predicates they are instantiated with.  Goldreich originally suggested a random predicate, with a slightly superconstant arity~$k$.  However  algorithmic attacks on random $\CSP(\calP)$ by Bogdanov and Qiao~\cite{BQ09} showed that predicates that are not at least ``$3$-wise uniform'' do not lead to secure PRGs with significant stretch.  Quite a few subsequent works have tried to analyze what properties of a predicate family~$\calP$ may --- or may not --- lead to secure PRGs~\cite{BQ09,ABR12, OW14, AL16}.

Regarding the approach of Daniely~et~al.\ to hardness of learning, there are close connections between the predicates for which random $\CSP(\calP)$ is assumed hard and the concept class for which one achieves hardness of learning.  For example, the earlier work~\cite{DLS14} assumed hardness of refuting random $\CSP(P^{\pm})$ for $P$ being (i)~the ``Huang predicate''~\cite{Hua13, Hua14}, (ii)~Majority, (iii)~a certain AND of~$8$ thresholds; it thereby deduced hardness of learning (i)~DNFs, (ii)~halfspaces with noise, (iii)~intersections of halfspaces.  Unfortunately, Allen et~al.~\cite{AOW15} gave efficient algorithms refuting all three hardness assumptions; fortunately, the results were mostly recovered in later works~\cite{DSS14,Dan15} assuming hardness of refuting random $k$-SAT and $k$-XOR.  Although these are more ``standard'' predicates, a careful inspection of~\cite{DSS14}'s hardness of learning DNF result shows that it essentially works by reduction from $\CSP(P^{\pm})$ where~$P$ is a ``tribes'' predicate. (It first shows hardness for this predicate by reduction from $k$-SAT.)  From these discussions, one can see the utility of understanding the hardness of random $\CSP(\calP)$ for as wide a variety of predicates~$\calP$ as possible.

\bfseries \item Superlinear number of constraints. \mdseries
Much of the prior work on hardness of refuting random CSPs (assumptions and evidence for it) has focused on the regime of $\Delta = O(1)$; i.e., random CSPs with $O(n)$ constraints.  However, it is quite important in a number of settings to have evidence of hardness even when the number of constraints is superlinear.  An obvious case of this arises in the application to security of Goldreich-style PRGs; here the number of constraints directly corresponds to the stretch of the PRG.  It's natural, then, to look for arbitrarily large polynomial stretch.  In particular, having $\mathsf{NC}^0$ PRGs with $m = n^{1 + \Omega(1)}$ stretch yields secure two-party communication with constant overhead~\cite{IKOS08}.  This motivates getting hardness of refuting random CSPs with $\Delta = n^{\Omega(1)}$.  As another example, the hardness of learning results in the work of Daniely~et~al.~\cite{DLS14, DSS14, Dan15} all require hardness of refuting random CSPs with $m = n^C$, for arbitrarily large~$C$.  In general, given a predicate family~$\calP$, it is interesting to try to determine the least~$\Delta$ for which refuting random $\CSP(\calP)$ instances at density~$\Delta$ becomes easy.

\bfseries
\item Stronger refutation. \mdseries
Most previous work on the hardness of refuting random CSPs has focused just on weak refutation (especially in the proof complexity community), or on {$\delta$-refutation} for arbitrarily small $\delta > 0$.  The latter framework is arguably more natural: as discussed in~\cite{Fei02}, seeking just weak refutation makes the problem less robust to the precise model of random instances, and requiring $\delta$-refutation for some $\delta > 0$ allows some more natural CSPs like $k$-XOR (where unsatisfiable instances are easy to refute) to be discussed.  In fact, it is natural and important to study $\delta$-refutation for \emph{all} values of~$\delta$.  As an example, given~$\calP$ it is easy to show that there is a large enough constant~$\Delta_0$ such that for any $\Delta \geq \Delta_0$ a random instance~$\calI$ of $\CSP(\calP)$ has $\Opt(\calI) \leq \mu_{\calP} + o(1)$, where $\mu_{\calP}$ is the probability a random assignment satisfies a random predicate $P \in \calP$.  Thus it is quite natural to ask for $\delta$-refutation for $\delta = 1- \mu_{\calP} - o(1)$; i.e., for an algorithm that certifies the \emph{true} value of $\Opt(\calI)$ up to~$o(1)$ (whp).  This is sometimes termed \emph{strong refutation}.  As an example, Barak and Moitra~\cite{BM16} show hardness of tensor completion based on hardness of strongly refuting random $3$-SAT with~$\Delta \ll n^{1/2}$.  In general, there is a very close connection between refutation algorithms for $\CSP(\calP)$ and approximation algorithms for $\CSP(\calP)$; e.g., hardness of $\delta$-refutation results for LP- and SDP-based proof systems can be viewed as saying that random instances are $1-\delta$ vs.\ $\mu_{\calP} + o(1)$ \emph{integrality gap} instances for $\CSP(\calP)$.

\bfseries
\item Hardness against superpolynomial time. \mdseries
Naturally, we would prefer to have evidence against superpolynomial-time refutation, or even subexponential-time refutation, of random $\CSP(\calP)$; for example, this would be desirable for cryptography applications.  This desire also fits in with the recent surge of work on hardness assuming the Exponential Time Hypothesis (ETH).  We already know of two works that use a strengthening of the ETH for random CSPs.    The first, due to Khot and Moshkovitz~\cite{KM16}, is a candidate hard Unique Game, based on the assumption that random instances of $\CSP(P^{\pm})$ require time $2^{\Omega(n)}$ to strongly refute, where $P$ is the $k$-ary ``Hadamard predicate''.  The second, due to Razenshteyn~et~al.~\cite{RSW16} proves hardness for the Weighted Low Rank Approximation problem assuming that refuting random $4$-SAT requires time $2^{\Omega(n)}$.  An even further interesting direction, in light of the work of Feige, Kim, and Ofek~\cite{FKO06}, is to find evidence against efficient \emph{nondeterministic} refutations of random CSPs.
\end{enumerate}

\noindent These discussions lead us to the following goal:
\begin{quotation}
\textbf{Goal:} \emph{For every predicate family~$\calP$, provide strong evidence for the hardness of refuting random instances of $\CSP(\calP)$, with the best possible tradeoff between number of constraints, refutation strength, and running time.}
\end{quotation}

The main theorem in this work, stated in Section~\ref{sec:our-result}, completely accomplishes this goal in the context of the Sum of Squares (SOS) method.  Before stating our results, we review this method, as well as prior results in the direction of the above goal.

\subsection{Prior results in proof complexity, and the SOS method}
Absent the ability to even prove $\PTIME \neq \NP$, the most natural way to get evidence of hardness for refuting random $\CSP(\calP)$ is to prove unconditional negative results for specific proof systems.  It's particularly natural to consider automatizable proof systems, as these correspond to efficient deterministic refutation algorithms.

Much of the work in this area has focused on random instances of $k$-SAT. A seminal early work of Chv\'{a}tal and Szemer\'{e}di~\cite{CS88} showed that Resolution refutations of random instances of $k$-SAT require exponential size when $\Delta$ is a sufficiently large constant.
Ben-Sasson and Wigderson~\cite{BSW01, Ben01} later strengthened this result to show that Resolution refutations require width $\Omega(\frac{n}{\Delta^{1/(k-2)+\eps}})$ for any $\eps > 0$.
Ben-Sasson and Impagliazzo and Alekhnovich and Razborov further extended these results to the Polynomial Calculus proof system~\cite{BI99, AR01a}; for example, the latter work showed that Polynomial Calculus refutations of random $k$-SAT instances with density~$\Delta$ require degree $\Omega(\frac{n}{\Delta^{2/(k-2)} \log \Delta})$.

On the other hand, much of the positive work on refuting random $k$-SAT has used spectral techniques and semialgebraic proof systems.  These latter proof systems are often automatizable using linear programming and semidefinite programming, and thereby have the advantage that they can naturally give stronger $\delta$-refutation algorithms.  As examples, Goerdt and Krivelevich~\cite{GK01} showed that spectral techniques (which can be captured by SDP hierarchies) enable refutation of random $k$-SAT with $m = n^{\lceil k/2 \rceil}$ constraints; Friedman and Goerdt~\cite{FG01} improved this to $m = n^{3/2+o(1)}$ in the case of random $3$-SAT.  One of the first lower bounds for random CSPs using SDP hierarchies was given by Buresh-Oppenheim~et~al.~\cite{BGH+03}; it showed that the \Lovasz--Schrijver$_+$ (LS$_+$) proof system cannot refute random instances of $k$-SAT with $k \geq 5$ and constant~$\Delta$.  Alekhnovich, Arora, and Tourlakis~\cite{AAT05} extended this result to random instances of $3$-SAT.

The strongest results along these lines involve the Sum of Squares (AKA Positivstellensatz or Lasserre) proof system.  This system, parameterized by a tuneable ``degree'' parameter~$d$, is known to be very powerful; e.g., it generalizes the degree-$d$ Sherali--Adams$_+$ (SA$_+$) and LS$_+$ proof systems.  In the context of CSP$(\calP)$ over domain~$\{0,1\}$, it is also (approximately) automatizable in $n^{O(d)}$ time using semidefinite programming.   As such, it has proven to be a very powerful positive tool in algorithm design, both for CSPs and for other tasks; in particular, it has been used to show that several conjectured hard instances for CSPs are actually easy~\cite{BBH+12, OZ13, KOTZ14}.  Finally, thanks to work of Lee, Raghavendra, and Steurer~\cite{LRS15}, it is known that constant-degree SOS approximates the optimum value of CSPs at least as well as \emph{any} polynomial-size family of SDP relaxations. See, e.g.,~\cite{OZ13, BS14, Lau09} for surveys concerning SOS.

Early on, Grigoriev~\cite{Gri01} showed that SOS of degree $\Omega(n)$ could not refute $k$-XOR instances on sufficiently good expanders.  Schoenebeck~\cite{Sch08} essentially rediscovered this proof and showed that it applied to random instances of $k$-SAT and $k$-XOR, specifically showing that SOS degree ~$\frac{n}{\Delta^{2/(k-2) - \eps}}$ is required to refute instances with density~$\Delta$.  Tulsiani~\cite{Tul09} extended this result to the alphabet-$q$ generalization of random $3$-XOR.

Much less was previously known about predicates other than $k$-SAT and $k$-XOR. Austrin and Mossel~\cite{AM08} established a connection between hardness of $\CSP(\calP)$ and pairwise-uniform distributions, showing inapproximability beyond the random-threshold subject to the Unique Games Conjecture. A  key work of Benabbas~et~al.~\cite{BGMT12} showed an unconditional analog of this result: random instances of $\CSP(P^{\pm})$ with sufficiently large constant constraint density require $\Omega(n)$ degree to refute in the SA$_+$ SDP hierarchy when $P$ is a predicate (over any alphabet) supporting a pairwise-uniform distribution on satisfying assignments.
O'Donnell and Witmer~\cite{OW14} extended these results by observing a density/degree tradeoff: they showed that if the predicate supports a $(t-1)$-wise uniform distribution, then the SA LP hierarchy at degree $n^{\Omega(\eps)}$ cannot refute random instances of $\CSP(P^{\pm})$ with $m = n^{t/2 - \eps}$ constraints.  They also showed the same thing for the SA${}_+$ SDP hierarchy, provided one can remove a carefully chosen~$o(m)$ constraints from the random instance.  Extending results of Tulsiani and Worah~\cite{TW13}, Mori and Witmer~\cite{MW16} showed this result for the SA$_+$ and LS$_+$ SDP hierarchies, for purely random instances. Finally, Barak, Chan, and Kothari~\cite{BCK15} recently extended the~\cite{BGMT12} result to the SOS system, though not for purely random instances:  they showed that for any Boolean predicate~$P$ supporting a pairwise-uniform distribution, if one chooses a random instance of $\CSP(P^{\pm})$ with large constant~$\Delta$ and then carefully removes a certain $o(n)$ constraints, then SOS needs degree~$\Omega(n)$ to refute the instance.

Beyond semialgebraic proof systems and hierarchies, even less is known about non-SAT, non-XOR predicates.  Feldman, Perkins, and Vempala \cite{FPV15} proved lower bounds for refutation of CSP$(P^\pm)$ using statistical algorithms when $P$ supports a $(t-1)$-wise uniform distribution.  Their results are incomparable to the above lower bounds for LP and SDP hierarchies: the class of statistical algorithms is quite general and includes any convex relaxation, but the \cite{FPV15} lower bounds are not strong enough to rule out refutation by polynomial-size SDP and LP relaxations.

\paragraph{Summary.} For the strongest semialgebraic proof system, SOS, our evidence of hardness for random CSPs from previous work was somewhat limited.  We did not know any hardness results for a superlinear number of constraints, except in the case of $k$-SAT/$k$-XOR and the alphabet-$q$ generalization of $3$-XOR.  We did not know any results that differentiated weak refutation from $\delta$-refutation.  Finally, the results known for refuting $\CSP(P^{\pm})$ with pairwise-uniform-supporting~$P$ did not hold for purely random instances.

\subsection{Our result} \label{sec:our-result}
We  essentially achieve the Goal described in Section~\ref{sec:desire} in the context of the powerful SOS hierarchy.  Specifically, for every predicate family~$\calP$, we provide a \emph{full three-way tradeoff between constraint density, SOS degree, and strength of refutation}.  Our lower bound subsumes all of the hardness results for semialgebraic proof systems mentioned in the previous section.  Furthermore, as we will describe,  known algorithmic work implies that our full three-way hardness tradeoff is tight, up to lower-order terms.

To state our result, we need a definition. For a predicate $P: \Omega^k \to \{0,1\}$ and an integer $1 < t \leq k$, we define $\delta_P(t)$ to be $P$'s distance from supporting a $t$-wise uniform distribution.  Formally,
\[
\delta_P(t) :=  \min_{\substack{\text{$\mu$  is a $t$-wise uniform distribution on $\Omega^k$,} \\ \text{$\sigma$ is a distribution supported on satisfying assignments for $P$}}} \dtv{\mu}{\sigma},
\]
where $\dtv{\cdot}{\cdot}$ denotes total variation distance.

We can now (slightly informally) state our main theorem in the context of Boolean predicates:
\begin{theorem} \label{thm:main-intro}
Let $P$ be a $k$-ary Boolean predicate and let $1 < t \leq k$.  Let $\calI$ be a random instance of $\CSP(P^{\pm})$ with $m = \Delta n$ constraints.  Then with high probability, degree-$\wt{\Omega}\left(\frac{n}{\Delta^{2/(t-1)}}\right)$ SOS fails to $(\delta_P(t)+o(1))$-refute $\calI$.
\end{theorem}

Additionally, in the case that $\delta_P(t) = 0$, our result does not need the additive~$o(1)$ in refutation strength. That is:
\begin{theorem}                                     \label{thm:main2}
Let $P$ be a $k$-ary predicate and let $\cmplx(P)$ be the minimum integer $3 \leq \tau \leq k$ for which $P$ fails to support a $\tau$-wise uniform distribution.  Then if $\calI$ is a random instance of $\CSP(P^{\pm})$ with $m = \Delta n$ constraints, with high probability degree-$\wt{\Omega}\left(\frac{n}{\Delta^{2/(\cmplx(P)-2)}}\right)$ SOS fails to (weakly) refute~$\calI$.
\end{theorem}
\begin{remark}
We comment here on the (surprisingly mild) parameter-dependence hidden by the $\wt{\Omega}(\cdot)$ and~$o(1)$ in these bounds.  See~Section~\ref{sec:parameters} for full details.
\begin{itemize}
    \item  In terms of $\Delta$, the $\wt{\Omega}(\cdot)$ is only hiding a factor of $\log \Delta$. Thus we get a full linear $\Omega(n)$-degree lower bound for $m = O(n)$ in both theorems above.
    \item In terms of $k$, and $t$, the $\wt{\Omega}(\cdot)$ is only hiding a factor of $1/(k 2^{O(k/t)})$. There are a number of interesting cases where one may take $t = \Theta(k)$; for example, $k$-SAT, $k$-XOR, and ${\mathrm{XOR}_{k/2} \oplus \mathrm{MAJ}_{k/2}}$, a predicate often used in cryptography (e.g., it was suggested by~\cite{AL16} for as the basis for high-stretch PRGs in $\mathsf{NC}^0$). In these cases, the dependence of the degree lower bound depends only \emph{linearly} on~$k$ and thus, there's little loss in having~$k$ significantly superconstant.
    \item Indeed in this case of $t = \Theta(k)$, if we also have $\Delta = 2^{\Theta(k)}$ then the degree lower bound for weak refutation in Theorem~\ref{thm:main2} is $\Omega(n)$ for $k$ as large as $\Omega(n)$; here, both $\Omega(\cdot)$'s hide only a \emph{universal} constants.  The regime of $\Delta = 2^{\Theta(k)}$ is  the algorithmically hardest one for $k$-SAT, and thus in this very natural case we have a linear-degree lower bound even for $k = \Omega(n)$.
    \item The refutation strength  $\delta_P(t)+o(1)$ in Theorem~\ref{thm:main-intro} is more precisely $\delta_P(t) + O(1/\sqrt{n})$ whenever $\Delta = n^{\Omega(1)}$.
    \item Theorem~\ref{thm:main-intro} also holds for predicates~$P$ with alphabet size~$q > 2$, with absolutely no additional parameter dependence on~$q$.
    \end{itemize}
\end{remark}

The full three-way tradeoff in Theorem~\ref{thm:main-intro} between constraint density, SOS degree, and strength of refutation is tight up to a polylogarithmic factor in the degree and an additive $o(1)$ term in the strength of the refutation. The tightness  follows from the below theorem, which is an immediate consequence of the general $\delta$-refutation framework of Allen et al.~\cite{AOW15} and the strong refutation algorithm for XOR due to Raghavendra, Rao, and Schramm~\cite{RRS16} (which fits in the SOS framework).
\begin{theorem} \textup{(Follows from \cite{AOW15, RRS16}.)} \label{thm:ub-intro}
Let $P$ be a $k$-ary Boolean predicate and let ${1 < t \leq k}$. Let $\calI$ be a random instance of $\CSP(P^\pm)$ with $m = \Delta n$ constraints.  Then with high probability, degree-$\wt{O}\left(\frac{n}{\Delta^{2/(t-2)}}\right)$ SOS \emph{does} $(\delta_P(t)-o(1))$-refute~$\calI$.  Furthermore, with high probability degree-$O(1)$ SOS succeeds in $(\delta_P(2)-o(1))$-refuting~$\calI$, provided $\Delta$ is at least some $\polylog(n)$.
\end{theorem}


\paragraph{An example.} As the parameters can be a little difficult to grasp, we illustrate our main theorem and its tightness with a simple example.  Let $P$ be the $3$-bit predicate that is true if \emph{exactly} one if its three inputs is true.  The resulting $3$-SAT variant $\CSP(P^\pm)$ is traditionally called $\text{$1$-in-$3$-SAT}$.  Let us compute the $\delta(t)$ values.  The uniform distribution on the odd-weight inputs is pairwise-uniform, and it only has probability mass $\frac14$ off of $P$'s satisfying assignments.  This is minimum possible, and therefore $\delta_{\text{$1$-in-$3$-SAT}}(2) = \frac{1}{4}$.  The only $3$-wise uniform distribution on $\{0,1\}^3$ is the fully uniform one, and it has probability mass $\frac58$ off of $P$'s satisfying assignments; thus $\delta_{\text{$1$-in-$3$-SAT}}(3) = \frac{5}{8}$.

Let us also note that as soon as $\Delta$ is a large enough constant, $\Opt(\calI) \leq \frac38 + o(1)$ (with high probability, a qualifier we will henceforth omit). Furthermore, it's long been known~\cite{BB02} that for~$\Delta=O(\log n)$ there is an efficient algorithm that weakly refutes $\calI$; i.e., certifies $\Opt(\calI) < 1$.  But what can be said about stronger refutation?  Let us see what our Theorem~\ref{thm:main-intro} and its counterpart Theorem~\ref{thm:ub-intro} tell us.

 Suppose first that there are $m = n\,\polylog(n)$ constraints. Theorem~\ref{thm:ub-intro} tells us that constant-degree SOS certifies $\Opt(\calI) \leq \frac34 + o(1)$.  However our result, Theorem~\ref{thm:main-intro}, says this~$\frac34$ cannot be improved: SOS cannot certify $\Opt(\calI) \leq \frac34 - o(1)$ until the degree is as large as $\wt{\Omega}(n)$.  (Of course at degree~$n$, SOS can certify the exact value of $\Opt(\calI)$.)

 What if there are $m = n^{1.1}$ constraints, meaning $\Delta = n^{.1}$?  Our result says SOS still cannot certify $\Opt(\calI) \leq \frac34 - o(1)$ until the degree is as large as $n^{.8}/O(\log n)$.  On the other hand, as soon as the degree gets bigger than some $\wt{O}(n^{.8})$, SOS \emph{does} certify $\Opt(\calI) \leq \frac34 - o(1)$; in fact, it certifies $\Opt(\calI) \leq \frac38 + o(1)$.

 Similarly (dropping lower-order terms for brevity), if there are $m = n^{1.2}$ constraints, SOS is stuck at certifying just $\Opt(\calI) \leq \frac34$ up until degree $n^{.6}$, at which point it jumps to being able to certify the truth, $\Opt(\calI) \leq \frac38 + o(1)$.  If there are $n^{1.49}$ constraints, SOS remains stuck at certifying just $\Opt(\calI) \leq \frac34$ up until degree $n^{.02}$.  Finally (as already shown in~\cite{AOW15}), once $m = n^{1.5}\,\polylog(n)$, constant-degree SOS can certify $\Opt(\calI) \leq \frac38 + o(1)$. \hfill \textbf{(End of example.)}\\

More generally, for a given predicate~$P$ and a fixed number of random constraints $m = n^{1+c}$, we provably get a ``time vs.\ quality'' tradeoff with an intriguing discrete set of breakpoints:  With constant degree, SOS can $\delta_P(2)$-refute, and then as the degree increases to $n^{1-2c}$, $n^{1-c}$, $n^{1-2c/3}$, etc., SOS can $\delta_P(3)$-refute, $\delta_P(4)$-refute, $\delta_P(5)$-refute, etc.

An alternative way to look at the tradeoff is by fixing the SOS degree to some $n^\eps$ and considering how refutation strength varies with the number of constraints.  So for $m$ between $n$ and $n^{3/2-\eps/2}$ SOS can $\delta_P(2)$-refute; for $m$ between $n^{3/2 - \eps/2}$ and $n^{2-\eps}$ SOS can $\delta_P(3)$-refute; for $m$ between $n^{2-\eps}$ and $n^{5/2 - 3\eps/2}$ SOS can $\delta_P(4)$-refute; etc.

It is particularly natural to examine our tradeoff in the case of constant-degree SOS, as this corresponds to polynomial time.  In this case, our Theorem~\ref{thm:main-intro} says that random $\CSP(P^\pm)$ cannot be $(\delta_P(t) + o(1))$-refuted when $m \ll n^{(t+1)/2}$, and it cannot even be weakly refuted  when $m \ll n^{\cmplx(P)/2}$.  Now by applying the work of Lee, Raghavendra, and Steurer~\cite{LRS15}, we get the same                                      hardness results for \emph{any} polynomial-size SDP-based refutation algorithm.  (See~\cite{LRS15} for precise definitions.)
\begin{corollary}                                       \label{cor:noSDP}
Let $P$ be a $k$-ary predicate,  and fix a sequence of polynomial-size SDP relaxations for $\CSP(P^\pm)$.  If $\calI$ is a random instance of $\CSP(P^\pm)$ with $m \leq \wt{\Omega}(n^{\cmplx(P)/2})$ constraints, then whp the SDP relaxation will have value~$1$ on~$\calI$.  Furthermore, if $m \leq \wt{\Omega}(n^{(t+1)/2})$ (for $1 < t \leq k$), then  whp the SDP relaxation will have value
at least $1 - \delta_P(t) - o(1)$ on $\calI$.
\end{corollary}
The results in this corollary are tight up to the polylogs on~$m$, by the SOS algorithms of~\cite{AOW15}.

\ignore{
\subsection{Junk}
In the case of $P = \text{$k$-SAT}$, much is known about both upper and lower bounds for refutation.  Chv\'{a}tal and Szemer\'{e}di proved that resolution refutations of random instances of $k$-SAT with $O(n)$ constraints require exponential size \cite{CS88};  Ben-Sasson and Wigderson later strengthened this result to show that resolution refutations of $k$-SAT require superpolynomial size even when $m = n^{3/2-\eps}$ \cite{BSW01}.  Ben-Sasson and Impagliazzo and Alekhnovich and Razborov proved the analogous theorem for polynomial calculus \cite{BI99, AR01a}.

More recent results have focused on static semialgebraic proof systems and their corresponding semidefinite programming (SDP) hierarchies.  Goerdt and Krivelevich showed the spectral techniques, which can be captured by SDP hierarchies, enable refutation of $k$-SAT with $m = n^{\lceil k/2 \rceil}$ \cite{GK01}.  For $3$-SAT, Friedman and Goerdt improved this to $m = n^{3/2+\eps}$, again using spectral methods \cite{FG01}.  One of the first lower bounds for refutation of $3$-SAT using a hierarchy was given by Alekhnovich, Arora, and Tourlakis for the \Lovasz-Schrijver (LS) proof system and instances with $O(n)$ \cite{AAT05}.  Buresh-Oppenheim et al. prove lower bounds for refuting random $k$-SAT with $m=O(n)$ using the stronger \Lovasz-Schrijver$_+$ (LS$_+$) SDP hierarchy \cite{BGH+03}.  Grigoriev gave lower bounds for refutation of random instances of XOR using the even stronger SOS hierarchy in a slightly different model \cite{Gri01}; Schoenebeck rediscovered his proof and showed that it applied to random instances of $k$-SAT and $k$-XOR in the model we consider here for $m$ up to $n^{k/2-\eps}$ \cite{Sch08}.

We extend this line of work in two directions.  First, we consider arbitrary predicates $P$.  In his breakthrough work, Feige showed worst-hardness of approximation results assuming hardness of refuting random $3$-SAT \cite{Fei02}.  Subsequent work has used assumptions about the hardness of refuting CSPs with other predicates to prove additional worst-case hardness results \cite{GL04, AAM+11}.  Furthermore, the recent work of Daniely and his coauthors has used this paradigm to prove hardness of learning results based on hardness of refuting CSPs with non-SAT predicates \footnote{While \cite{AOW15} invalidates several of the hardness-of-refutation assumptions for non-SAT predicates used in \cite{DLS14}, the technique remains promising.  In addition, even though \cite{Dan15} only assumes hardness of refuting SAT, much of the work in the paper is devoted to deriving hardness of refuting of a non-SAT predicate from hardness of refuting SAT.} \cite{DLS13, DLS14, DSS14, Dan15}.    Barak, Kindler, and Steurer used the stronger assumption that a basic SDP optimally refutes any CSP to prove a couple of hardness of approximation results \cite{BKS13}.

Beginning with a paper of Goldreich \cite{Gol00}, pseudorandom generators (PRGs) based on applying a predicate $P$ to random subsets of output bits have been extensively studied because they can be computed in constant parallel time \cite{ABR12, ABW10, AL16}.  CHECK IF ABW FALLS INTO THIS?  See Applebaum's survey for details \cite{App13b}.  Hardness of refuting $\CSP(P)$ provides evidence that these PRGs are secure and proving such results for larger $m$ gives evidence that these PRGs have high stretch.  On the algorithmic side, Allen et al. showed that a constant-degree SOS algorithm can refute instances of $\CSP(P)$ with $m \gg n^{\cmplx(P)/2}$ constraints, where $\cmplx(P)$ measures the complexity of $P$ \cite{AOW15}.  Specifically, $\cmplx(P)$ is the smallest integer $t$ for which there is no $t$-wise uniform distribution supported on satisfying assignments to $P$.  Raghavendra, Rao, and Schramm extended this result to higher degree, showing that an $O(n^{2\eps/(\cmplx(P)-2)})$-degree SOS algorithm can refute instances of $\CSP(P)$ with $m \gg n^{\cmplx(P)/2-\eps}$ constraints \cite{RRS16}.

Second, we study a stronger notion of refutation.  Not only is random $3$-SAT unsatisfiable when $m \gg n$, at most a $\frac{7}{8}+o(1)$ fraction of constraints can be simultaneously satisfied.  For $\CSP(P)$, the maximum fraction of satisfiable constraints is $\frac{|P^{-1}(0)|}{2^k} + o(1)$ when $m \gg n$.  With this in mind, we consider the $\delta$-refutation problem, in which we wish to find a certificate that at most a $1-\delta$ fraction of constraints can be simultaneously satisfied.  We will call $\delta$ the strength of the refutation.  The problem of certifying that at most a $\frac{|P^{-1}(0)|}{2^k} + o(1)$ fraction of constraints can be satisfied ($\delta = 1-\frac{|P^{-1}(0)|}{2^k}-o(1)$) is called strong refutation.  Both Feige's and Daniely at al.'s hardness results actually assume hardness of $\delta$-refutation for some $\delta$ CHECK THIS.  These stronger notions of refutation has also been studied in previous algorithmic work \cite{COGL04, COCF10, BM15, AOW15}.  Of course, lower bounds for refutation are lower bounds for $\delta$-refutation for any $\delta > 0$ GET THIS RIGHT.  However, even at larger numbers of constraints for which refutation is easy, there may be values of $\delta$ for which $\delta$-refutation remains hard.  We are not aware of any previous work giving lower bounds for $\delta$-refutation or strong refutation at densities higher than those for which lower bounds for refutation are known.  Consider the following parameter $\delta_P(t)$.
\[
\delta_P(t) :=  \min_{\substack{\text{$t$-wise uniform distributions $\mu$ on $\{0,1\}^k$} \\ \text{distributions $\sigma$ supported on satisfying assignments to $P$}}} \dtv{\mu}{\sigma}.
\]
Allen et al. show that their constant-degree SOS algorithm $(\delta_P(t) - o(1))$-refutes instances of $\CSP(P)$ with $m \gg n^{t/2}$ \cite{AOW15}.  By taking $t = k$, they strongly refute when $m \gg n^{k/2}$.  Again, Raghavendra, Rao, and Schramm's result can be used to show that an $O(n^{2\eps/(t-2)})$-degree SOS algorithm can $(\delta_P(t)-o(1))$-refute instances of $\CSP(P)$ with $m \gg n^{t/2-\eps}$ constraints \cite{RRS16}.  Finally, in the context of hierarchies, it is natural to ask how SOS degree increases as a function of refutation strength.  Chlamtac and Singh ask a similar question for the $3$-uniform hypergraph independent set problem \cite{CS08} and show that approximation factor increases as a function of number of rounds in the Sherali-Adams$_+$ SDP hierarachy.  To our knowledge, no similar results were known for CSPs or for the SOS hierarchy.

\paragraph{Our result} We prove lower bounds for $\delta$-refutation of $\CSP(P)$ using SOS.
\begin{theorem} \label{thm:main-intro}
FILL THIS IN
\end{theorem}
For all degrees, this result matches the upper bound of \cite{RRS16} up to a multiplicative polylogarithmic factor in the number of constraints and up to an additive $o(1)$ in the strength of the refutation.  In other words, it gives an almost exact tradeoff between the number of constraints and SOS degree required to $\delta$-refute any CSP.  We can state this tradeoff in a couple of ways.  Fixing $m = n^{c+1}$, we consider SOS degree and refutation strength as a function of the parameter $t$.  Our result and that of \cite{RRS16} say that, to within polylogarithmic factors, SOS degree $n^{\frac{t-2-2c}{t-2}}$ is both necessary and sufficient to $\delta_P(t)$-refute instances of $\CSP(P)$.  Refutation strength increases in discrete steps as the degree increases: we can $\delta_P(3)$-refute once the degree is as least $n^{1-2c}$, we can $\delta_P(4)$-refute once the degree is at least $n^{1-c}$, etc.  We can also set SOS degree to $n^d$ and consider $m$ and refutation strength as a function of $t$ to see that $m = n^{t/2(1-d)+d}$ constraints are necessary and sufficient for $\delta_P(t)$-refutation.  Now, refutation strength increases in discrete steps as $m$ increases: we can $\delta_P(3)$-refute once $m$ is at least $n^{3/2-d/2}$, we can $\delta_P(4)$-refute once $m$ is at least $n^{2-d}$, etc.

Via work of Lee, Raghavendra, and Steurer \cite{LRS15}, our result together with that of \cite{AOW15} also gives an almost exact tradeoff between number of constraints and strength of refutation for \emph{any} SDP-based refutation algorithm of the following form.
\begin{enumerate}
\item Solve a polynomial-size SDP whose objective function estimates the fraction of satisfied constraints.
\item If $\SDPOpt \leq 1-\delta$, return ``not $(1-\delta)$-satisfiable".  Else, return ``fail".
\end{enumerate}
In particular, up to polylogarithmic factors, $m=n^{t/2}$ constraints are both necessary and sufficient for $\delta_P(t)$-refutation in this model.  As far as we know, such SDP-based algorithms capture all state-of-the-art, polynomial-time refutation algorithms except for Gaussian elimination, which can be used to refute $k$-XOR instances \cite{AOW15}.  Additionally, we believe that this model captures all known polynomial-time $\delta$-refutation algorithms with constant $\delta$ \cite{AOW15}.  Looking beyond average-case instances, an SDP algorithm gives the best-possible polynomial-time approximation to the optimal value of any CSP in the worst case, assuming the Unique Games Conjecture \cite{Rag08}.

\paragraph{Comparison to previous work} O'Donnell and Witmer proved an analogous result for the Sherali-Adams (SA) linear programming hierarchy \cite{OW14} with $\delta = 0$ FIX THIS based on techniques of Benabbas et al. \cite{BGMT12}.  Mori and Witmer showed the corresponding lower bound for the Sherali-Adams$_+$ (SA$_+$) and \Lovasz-Schrijver$_+$ (LS$_+$) SDP hierarchies using methods of Tulsiani and Worah \cite{TW13}, again only for $\delta=0$.

For SOS, Barak, Chan, and Kothari proved that random instances of $\CSP(P)$ with $P$ supporting a pairwise-uniform distribution over satisfying assignments, $m = O(n)$, and $o(m)$ constraints removed cannot be refuted by SOS with degree $o(n)$ \cite{BCK15}.  Our result improves that of Barak et al. in several ways.  First, we show that removing $o(m)$ constraints is not necessary.  Our lower bound holds in the more natural and extensively-studied setting of refuting fully random instances.  Second, our result holds for much larger values of $m$.  Hardness of refutation with larger numbers of constraints is important because of connections to security of local PRGs and hardness of learning described above.  Third, we generalize from CSPs with predicates supporting pairwise-uniform distributions to those with predicates supporting $t$-wise uniform distributions.  In both the PRG and hardness-of-learning applications, one chooses a predicate based on the statements one wants to prove; this freedom to choose a predicate with helpful properties is crucial.  Fourth, we show lower bounds for $\delta$-refutation for any $\delta \in [0, 1-\frac{|P^{-1}(0)|}{2^k}-o(1)]$.  Increasing $\delta$ makes the problem harder and our lower bounds become correspondingly stronger.  As mentioned earlier, many hardness results based on refutation actually assume hardness of $\delta$-refutation for some $\delta > 0$.  Finally, we believe our proof gives new conceptual insight into construction of SOS lower bounds that we hope will be useful in proving SOS lower bounds for other problems.  We discuss this further in the next section.
}

\section{Technical framework}
In Section~\ref{sec:intro}, we described our results as being SOS lower bounds for random CSPs, with constraints chosen randomly from a fixed predicate family~$\calP$.  However it is conceptually clearest to divorce our results from the ``random CSP'' model as quickly as possible.
\begin{itemize}
    \item Our lower bound applies whenever the underlying factor graph (bipartite constraint/variable graph) does not contain certain small forbidden subgraphs, which we call ``implausible'' subgraphs.  Granted, the only examples we know of such graphs are random graphs (whp).  Further, the condition of ``does not contain any implausible subgraphs'' is highly related to the condition of ``has very good vertex expansion''. Still, we believe the right way to think about the requirement is in terms of forbidden subgraphs.
    \item Our lower bound doesn't really involve CSPs and constraints, per se.  For each constraint-vertex~$f$ in the underlying factor graph, rather than assuming it comes equipped with a constraint predicate~$P$ applied to its vertex-variable neighbors, we assume it comes equipped with a probability distribution~$\mu_f$ on assignments to its vertex-variable neighbors.  We can have a different $\mu_f$ for every constraint-vertex~$f$ if we want (indeed, the constraints need not even have the same arity).
    \item Our SOS lower bounds now take the following form:  Assume we are given a factor graph~$G$ with no implausible subgraphs, and assume each constraint-vertex~$f$ has an associated distribution $\mu_f$ that is $t$-wise uniform.  Then the low-degree SOS proof system ``thinks'' that there is a global assignment to the variables such that, at every constraint-vertex~$f$, the local assignment to the neighboring variable-vertices is in the support of~$\mu_f$. (Indeed, it ``thinks'' that there is a \emph{probability distribution} on global assignments such that for almost all~$f$, the marginal distribution on $f$'s neighbors is equal to~$\mu_f$.)
\end{itemize}

Let us make some of these notions more precise.

\subsection{Constraint satisfaction} \label{sec:constraint-satisfaction}

\begin{notation}
    We fix an \emph{alphabet} $\alphabet$ of cardinality~$\alphasize \geq 2$, and  a \emph{maximum constraint arity} $\maxarity \geq 3$.
\end{notation}
The reader is strongly advised to focus on the case $\alphasize = 2$, with $\alphabet = \{\pm 1\}$, as the only real difficulty posed by larger alphabets is notational.   Also, although we describe $\maxarity$ as a maximum arity, there will be no loss in thinking of every constraint as having arity~$K$.
\begin{definition}[$t$-wise uniform distributions]
    A probability distribution $\mu$ on $\alphabet^k$ is said to be \emph{$t$-wise uniform} if its marginal on every subset of $t$ coordinates is uniform.
\end{definition}
Rather than our full Theorem~\ref{thm:main-intro} concerning $\delta$-refutation, the reader is advised to mainly keep in mind our Theorem~\ref{thm:main2}, which is concerned with (weak) refutation of CSPs for which the predicates support a $(\tau-1)$-wise uniform distribution.  Given our proof of Theorem~\ref{thm:main2}, the more general Theorem~\ref{thm:main-intro} will fall out fairly easily.
\begin{notation}
    We fix an integer $\Tcxty$ satisfying $3 \leq \Tcxty \leq \maxarity$.
\end{notation}
The reader is advised to focus on the simplest case of $\Tcxty = 3$ (corresponding to predicates supporting \emph{pairwise}-uniform distributions), as the value of $\Tcxty$ makes no real difference to our proofs.
\begin{notation}[Instance]
    The \emph{instance} we work with consists of two parts: a \emph{factor graph} and its \emph{constraint distributions}.  The factor graph, denoted~$G$,  is a bipartite graph with edges going between $n$ \emph{variable-vertices} and $m$ \emph{constraint-vertices}. For a constraint-vertex~$f$ we write~$N(f)$ for the \emph{neighborhood} of~$f$, which we take to be an \emph{ordered} list of the variable-vertices adjacent to~$f$.  We assume that the degree (``arity'') of every constraint-vertex~$f$ satisfies $\Tcxty-1 \leq |N(f)| \leq \maxarity$.  Finally, each constraint-vertex~$f$ also comes with a constraint distribution $\mu_f$ on $\alphabet^{N(f)}$.  It is assumed that each $\mu_f$ is $(\Tcxty-1)$-wise uniform.
\end{notation}

\ignore{
The above set up corresponds to the following constraint satisfaction problem. We have a $m$ constraints $c_1, c_2, \ldots, c_m$ corresponding to each constraint vertex in $G$ over $n$ variables $x_1, x_2, \ldots, x_n$ corresponding to each variable vertex in $G$. Each $c_i$ is described by a tuple of size $|c_i|$ at most $\maxarity$ on $x_1, x_2, \ldots, x_n$ - the neighbors of the constraint vertex corresponding to $c_i$ in $G$ - and has an associated predicate $P_i: \alphabet^{|c_i|} \rightarrow \zo$. The algorithmic task is to come up with an assignment from $\alphabet$ to each $x_i$ such that the number of $P_i$ that take the value $1$ (i.e. are \emph{satisfied}) is maximized. We will deal with the task of finding if a given instance of CSP problem as above is satisfiable. This can be equivalently formulated as the following \emph{polynomial feasibility} (i.e. checking if there's a simultaneous solution to polynomial equality/inequality constraints) problem:

For each $i \in \vbls(G)$ and each $c \in \Omega$, we introduce an ``indeterminate'' $\indet{c}{i}$ that is supposed to stand for~$1$ if variable~$i$ is assigned~$c$ and~$0$ otherwise. For each constraint $c_i$, let $f_i$ be the degree $\maxarity \cdot \alphasize$ polynomial that evaluates to $1$ at an assignment $a$ from $\alphabet^{|c_i|}$ if $P_i$ is satisfied by $a$ and $0$ otherwise. Then, for CSP problem above can be formulates as finding if the following system of polynomial equations is feasible.

\begin{align}
  f_i = 1 & \forall \text{ }  i \\
\text{ s.t. } & \indet{c}{i}^2 = \indet{c}{i} \text{ } \forall \text{ } c,i. \label{eq:csp-polyopt-formulation}
\end{align}
}

To orient the reader vis-\`{a}-vis our description of CSPs in Section~\ref{sec:csps}, consider our Theorem~\ref{thm:main2} in which we have $\CSP(P^{\pm})$ instances, where $P : \{\pm 1\}^k \to \{0,1\}$ is a $k$-ary Boolean predicate with complexity $\cmplx(P) = \Tcxty$.  This means there exists some $(\tau-1)$-wise uniform distribution~$\mu$ on $\{\pm 1\}^k$ supported on satisfying assignments for~$P$.  Note that for any ``literal pattern'' $\ell \in \{\pm\}^{k}$, the distribution $\mu_\ell$ gotten by negating inputs to~$\mu$ according to~$\ell$ is also $(\Tcxty-1)$-wise uniform.  In the $\CSP(P^{\pm})$ instance, to every constraint with literal pattern $\ell$ the associated ``constraint distribution'' will be~$\mu_\ell$. (In the more general context of Theorem~\ref{thm:main-intro} where we have a $k$-ary predicate~$P$ with $\delta = \delta_P(t)$, this means there is some distribution $\mu$ on $\{\pm 1\}^k$ which is $t$-wise uniform and which is $\delta$-close to being supported on~$P$.  We will take $\Tcxty = t+1$ and take the constraint distributions to be $\mu_\ell$ again.)



\subsection{Plausible factor graphs} \label{sec:plausible-fg}

As mentioned earlier, our SOS lower bounds will hold whenever the factor graph~$G$ has no ``implausible'' subgraphs. The meaning of this will be discussed in much greater detail in Section~\ref{sec:plausible}, but here we will give the briefest possible definition.

\begin{notation}
    We introduce two parameters:  $1 \leq \CONSMALL \leq n/2$ and $0 < \price < 1$.  (For the sake of intuition, the reader might think of, e.g., $\CONSMALL = n^{\Omega(1)}$ and $\price = \frac{1}{\log n}$.)  The parameters are assumed to satisfy $\maxarity \leq \price \cdot \CONSMALL$.
\end{notation}
\paragraph{Plausibility Assumption.}\emph{Henceforth the factor graph $G$ is assumed to satisfy the following property: Let $H$ be an edge-induced subgraph in which every constraint-vertex has minimum degree~$\tau$.  Suppose $H$ has $c$~constraint-vertices, $v$~variable-vertices, and $e$~edges, with $c \leq 2 \cdot \CONSMALL$.  Then $(\tau - \price)c \geq 2(e-v)$.}\\

We call the subgraphs $H$ for which the inequality holds \emph{plausible} because they are indeed the ones that may plausibly show up when the factor graph~$G$ is randomly chosen:
\begin{proposition} \label{prop:randgraph} (Roughly stated; see Theorem~\ref{thm:random-graph} for a precise statement.)
    A random  $G$ with constraint density~$\Delta$ will satisfy the Plausibility Assumption whp provided $\displaystyle \CONSMALL \ll \frac{n}{\Delta^{2/(\Tcxty - 2 - \price)}}$.
\end{proposition}
The Plausibility Assumption is highly similar to the assumption that $G$ has good vertex-expansion, and indeed our proof of  Theorem~\ref{thm:random-graph} in Appendix~\ref{app:expansion} is a completely standard variant of the well-known proof that random bipartite graphs have good vertex-expansion.

\subsection{The Sum of Squares algorithm, and pseudoexpectations}
We give a brief overview of the Sum of Squares algorithm/proof system here.  For more general background see, e.g.,~\cite{BS16}; for more details germane to this paper, see Section~\ref{sec:pE}.

The Sum of Squares (SOS) algorithm is a hierarchy of semidefinite programming-based relaxations applicable to \emph{polynomial optimization problems}; i.e., maximizing an $n$-variate polynomial subject to polynomial inequality and equality constraints. Each algorithm in the hierarchy is indexed by a parameter~$d$ known as the \emph{degree} of the relaxation. Central to the algorithm is the concept of \emph{pseudoexpectations} that describe the feasible points of the SOS algorithm of degree $d$.

\begin{definition}[Pseudoexpectations]
Given $n$ indeterminates, a degree-$d$ pseudoexpectation is a linear operator~$\pE$ on the space of real polynomials of degree at most~$d$ in those indeterminates, such that $\pE[1] = 1$.  We also generally want it to satisfy the \emph{Positive Semidefiniteness} condition: $\pE[p^2] \geq 0$ for every polynomial $p$ of degree at most $d/2$.
\end{definition}
\begin{definition}[Pseudoexpectations satisfying an identity]
	A degree-$d$ pseudoexpectation~$\pE$ is said to \emph{satisfy a polynomial identity} ``$p = 0$'' if,  for every polynomial $q$ with $\deg(p) + \deg(q) \leq d$, we have $\pE[pq] = 0$.
\end{definition}


Given a polynomial optimization problem --- say, maximizing a polynomial $p_1$ subject to constraints $\{q_i = 0 : i \in [m]\}$ --- the degree-$d$ SOS relaxation maximizes $\pE[p_1]$ over all degree-$d$ pseudoexpectations $\pE$ that satisfy the identities $\{q_i = 0 : i \in [m]\}$.  A feasibility problem, in particular, would ask if there is a degree-$d$ pseudoexpectation satisfying certain  polynomial equality constraints. These SOS relaxations can be expressed using a semidefinite program (SDP) of size $n^{O(d)}$.  The \emph{Sum of Squares algorithm} refers to (approximately) solving the SDP, which can generally be done in $n^{O(d)}$ time.

As suggested by the name, pseudoexpectations generalize the notion of expectations with respect to a \emph{probability distribution} on real indeterminate values satisfying the given polynomial identity  constraints.  In particular, if there is at least one real solution for the polynomial identity constraints, then \emph{any} probability distribution on solutions yields a valid degree-$d$ pseudoexpectation, for any~$d$. However, even when the polynomial constraints have no real solution, there may well be  pseudoexpectations of limited degree that satisfy all the constraints.  As one would expect, as the degree~$d$ grows, the pseudoexpectations resemble actual expectations more and more. Indeed, if the constraints include that the $n$ indeterminates are Boolean (``$x_i^2 = x_i$'' or ``$x_i^2 = 1$'') then every degree-$2n$ pseudoexpectation in fact corresponds to an actual distribution on real solutions.

In our context of CSPs, we can think of a constraint satisfaction problem $\calE = \{(P_i, S_i)\}$ over~$n$ Boolean variables $x_1, \dots, x_n$ as a polynomial feasibility problem, with (the arithmetization of) the constraints $P_i(x_{S_i}) = 1$ as polynomial identities.  As we know, randomly chosen CSPs with $\Delta \gg 1$ are unsatisfiable whp; to show a lower bound on the degree-$d$ SOS refutation algorithm amounts to showing that there exists a degree-$d$ pseudoexpectation that satisfies all the constraints.  In more casual terminology, we say that degree-$d$ SOS ``thinks'' that the CSP is satisfiable.


%
\subsection{Main result}
We can now describe our main result with the terminology and set-up developed above. 

\begin{theorem} [Roughly stated; cf.~Theorem~\ref{thm:main}.] \label{thm:technical-approx}
Suppose we are given an instance, with factor graph~$G$ satisfying the Plausibility Assumption, and constraint distributions~$\mu_f$ for each constraint-vertex.  Then for $D = \frac13 \price \cdot \CONSMALL$, there exists a degree-$D$ pseudoexpectation~$\pE$ on global variable assignments such that for every constraint-vertex~$f$, the following (suitably encoded) polynomial identity is satisfied: ``The marginal distribution on assignments to the variable-neighbors of~$f$ is supported within $\supp(\mu_f)$.''  (Indeed, for almost all~$f$, a stronger identity is satisfied, that the marginal simply \emph{equals}~$\mu_f$.)
\end{theorem}
In particular, if our instance comes from an actual random CSP with predicates, where for each~$f$ the distribution $\mu_f$ is supported on satisfying assignments for the predicate at~$f$, then the degree-$D$ SOS algorithm ``thinks'' that the CSP is completely satisfiable. This is of course despite the fact that, whp, the CSP is not satisfiable.


Given Proposition~\ref{prop:randgraph} and Theorem~\ref{thm:technical-approx}, we can now point out how the constraint density vs.~SOS-degree tradeoff arises in our Theorem~\ref{thm:main2}.  For $\CSP(P^{\pm})$ with $\cmplx(P) = \tau$ and $\Delta n$ random constraints, we get an SOS lower bound for degree roughly $\price \cdot \frac{n}{\Delta^{2/(\Tcxty - 2 - \price)}}$.  The best choice of $\price$ is roughly $1/\log \Delta$, and this indeed yields a degree bound of $\wt{\Omega}\left(\frac{n}{\Delta^{2/(\cmplx(P)-2)}}\right)$.  More precise details of parameter-setting are given in Section~\ref{sec:parameters}.

\section{Sketch of our techniques}
Throughout this section, we describe our techniques in the context of CSPs on $n$ Boolean variables and $k$-ary predicates that are $(\tau-1)$-wise uniform. As stated before, almost all of our ideas are present in this special case. Our goal is to build a degree-$d$ pseudoexpectation operator $\pE$ as described in Theorem~\ref{thm:technical-approx}.
\subsection{Constructing the pseudoexpectation}
As in all previous works on CSP lower bounds for hierarchies, we use a variant of the natural pseudoexpectation introduced by Benabbas~et~al.~\cite{BGMT12}. This pseudoexpectation is always defined in terms of a certain ``closure'' operator on instance graphs; previous works have used slightly different notions of ``closure''.  Our method introduces yet another definition of closure that we believe is the ``right'' one; at the very least, it seems to be precisely the right definition for facilitating our proofs.



\subsubsection{Closures}
We can describe a pseudoexpectation by prescribing its values on the basis of monomials of degree at most~$d$.  We work with the Fourier basis; i.e., $\pm 1$ notation.

In the context of CSPs, a natural way to come up with a pseudoexpectation is via the idea of \emph{local distributions.} If $\pE$ is a degree-$d$ pseudoexpectation, then for every collection $S$ of at most $d/2$ variables, $\pE$ agrees with the expectation of an actual probability distribution. In particular, the pseudoexpectation of a monomial $x^S := \prod_{i \in S} x_i$ for $S \subseteq [n]$ (or indeed any function on $S$) can then be described as the expectation of $x^S$ with respect to the local distribution $\eta_S$ that $\pE$ induces on the set $S$ of variables. For such a definition to make sense, the local distributions must satisfy \emph{consistency}: the pseudoexpectation of $x^T$ should equal the expectation of $x^T$ with respect to the local distribution $\eta_S$ for any $S$ that includes $T$ and is of size at most $d$.


We would like to choose local distributions $\eta_S$ that are supported on satisfying assignments of all constraints completely included in $S$ (we call these the constraints covered by $S$). At first blush, we could choose the uniform distribution over the set of satisfying assignments for the constraints covered by $S$. However, this choice doesn't satisfy the consistency constraints. The $t$-wise uniform distributions that are supported on satisfying assignments of the predicate $P$ now come to our rescue: if we obtain a local probability distribution that induces $\mu$ on the literals of any constraint in our CSP instance, we should intuitively expect be in good shape because $t$-wise uniformity roughly guarantees that any constraint that intersects $S$ in $t$ or less variables has a satisfying assignment that agrees with the assignment sampled for $S$. A natural choice is to define the probability of an assignment to $S$ to be the product of the probabilities (with respect to~$\mu$) of the partial assignments corresponding to the constraints covered by $S$. This doesn't work as-is, either: there could be constraints  that intersect~$S$ in many variables and yet are not completely contained inside~$S$. A sample from $\eta_S$ thus might already force such a constraint to not be satisfied.

To correct for this, we want to collect all such ``dependencies'' before choosing the local distribution. Benabbas~et~al.~\cite{BGMT12} make this idea precise by defining a notion of \emph{closure} for a set of variables $S$: intuitively, these are all the variables that one should care about when defining the local distribution on $S$. Concretely, their closure maps $S$ into a larger set $S'$ such that for any $T \supseteq S'$, the marginal of $\eta_T$ on $S$ is equal to the marginal of $\eta_{S'}$ on $S$. We then choose $\eta_{S'}$ to be the local distribution on $S'$ and define $\eta_S$ to be the marginal of $\eta_{S'}$ on $S$. For such an effort to be feasible, $S'$ shouldn't be much bigger than $S$: if in the extreme case the closure happened to be the whole set of variables $[n]$, we cannot define a distribution on satisfying assignments of all constraints covered by $S'$.


The closure of Benabbas et al.~\cite{BGMT12} guarantees local consistency as we wanted.  Local consistency is all that is required for showing a Sherali--Adams lower bound and is equivalent to the following local positivity condition, which is weaker than positive semidefiniteness: $\pE[p] \geq 0$ for $p$ for every truly nonnegative polynomial~$p$ depending on at most~$d$ variables.  However, when trying to show that the more global $\pE[p^2]$ positive-semidefiniteness condition holds, the \cite{BGMT12} construction seems hard to analyze.

To address this problem, Barak, Chan, and Kothari \cite{BCK15} introduced a simpler variant of the \cite{BGMT12} closure in order to show that the $\pE$ defined above satisfies the positive-semidefiniteness condition for certain pruned random instances of the CSP$(P^{\pm})$, when $P$ supports a pairwise-uniform distribution. However, their definition of closure degenerates into the set of all variables with high probability when the random CSP has $\density = \omega(1)$.

\paragraph{Our closure.} One of the main innovations in our work is the introduction of a new, simpler definition of closure that plays a key role in our proof of positive semidefiniteness and gives a definition of $\pE$ that works even when the number of constraints is superlinear in $n$. In addition, our definition of closure enables us to extend our results to $\delta$-refutation. 

Our closure for a set of variables $S$ is a subgraph of the factor graph of the CSP instance, including both variables and constraints. We think of the closure of $S$ as being the set of variables and constraints that ``matter" when defining the distribution $\eta_S$.  Given that a predicate $P$ supports a $(\tau-1)$-wise uniform distribution, any constraint that affects $\eta_S$ must have at least $\tau-1$ variables in $S$.  Otherwise, $(\tau-1)$-wise uniformity implies that we could ignore such a constraint without changing $\eta_S$.  Any variable $v$ not in $S$ that occurs in only one constraint isn't necessary for defining~$\eta_S$, either.  We could sum $\eta_S$ over the two assignments to $v$ to get a new distribution that no longer depends on $v$.  This leads to a natural choice of the closure as the union of all small subgraphs of the factor graph such that each constraint contains at least $\tau-1$ variables and each variable outside of $S$ occurs in at least two constraints. For a formal definition, see Section \ref{sec:PE}.

\subsection{Proving positivity}
Once we have the definition of the pseudoexpectation, we get to the main challenge in showing any SOS lower bound: arguing positive-semidefiniteness of the $\pE$ constructed. The high level idea in our analysis builds on the work of Barak, Chan and Kothari~\cite{BCK15}. Their idea of proving positive-semidefiniteness is simple. They begin by observing that it suffices to verify positive-semidefiniteness for a basis that satisfies \emph{orthogonality} under $\pE[\cdot]$, meaning, the pseudoexpectation of the product of any distinct pair of basis polynomials is~ $0$.

\begin{fact} \label{fact:gs-local}
Suppose there exists a basis $f_1,f_2,\ldots$ for degree-$d$ polynomials such that the following two properties hold:
\begin{enumerate}
\item $\pE[f_i f_j] = 0$ for all $i \ne j$. \label{enum:orth}
\item $\pE[f_i^2] \geq 0$ for all $i$. \label{enum:pos}
\end{enumerate}
Then $\pE[g^2] \geq 0$ for all $g$ of degree at most $d$.
\end{fact}
\begin{proof}
Write $g$ as $\sum_i a_i f_i$.  Then
$\displaystyle
\pE[g^2] = \sum_{i,j} a_i a_j \pE[f_i f_j] = \sum_i a_i^2 \pE[f_i^2] \geq 0. \qedhere
$
\end{proof}
Notice that the standard Fourier monomial basis guarantees us positivity (since $\pE$ satisfies the local Sherali--Adams positivity condition by construction). However, it is not orthogonal in general. How can we construct such a basis? One way to construct a basis that is orthogonal under $\pE[\cdot]$ is to perform the Gram--Schmidt process on, say, the monomial basis $1,x_1,x_2,\ldots,x_1x_2,\ldots$ to get a new basis $f_1,f_2,\dots$.  Now, Property~\ref{enum:orth} above holds for this new basis by construction. However, the Gram--Schmidt process is highly sequential and, in particular, the basis function towards the end could depend on all~$n$ variables. Thus, we cannot appeal to local positivity of $\pE$ in order to argue positive-semidefiniteness of the newly generated basis. It appears that we have made no progress, ensuring orthogonality but potentially losing positivity.

The idea of Barak~et~al.~to escape this pitfall is to show that \emph{local orthogonalization} is enough.  Before the start of the Gram--Schmidt process, we fix an order on basis vectors.   In each step of the process, one orthogonalizes a basis function against all previous basis functions in this order by subtracting off its projection onto their span.  Barak~et~al.~analyze the variant of this process in which one orthogonalizes a basis function~$x^S$ by subtracting off its projection onto the span of all basis functions the precede it in the order \emph{and} are functions of variables that lie in a small ``ball'' around $S$ in the factor graph $G$ of the instance. This lets them ensure that the new basis satisfies positivity (since it now depends only on a small number of variables, one can appeal to the local positivity of $\pE$), and they show that this relaxed variant of the Gram--Schmidt process still ensures orthogonality.

Their proof, however, is highly combinatorial and requires various assumptions on the factor graph of the instance that intuitively shouldn't matter. In particular, they need that the factor graph have no small cycles (girth should be logarithmic): while this can be ensured by pruning $o(n)$ fraction of the constraints in a random instance with $\Theta(n)$ constraints, this proof strategy breaks down for super-linear number of constraints .


\paragraph{Our approach} Our main idea simplifies the analysis without requiring the assumptions of \cite{BCK15} and yields tight results. It also naturally extends to the case of $t$-wise uniform predicates and further to $\delta$-approximate $t$-wise uniform predicates. We next describe our key technical ideas that makes this possible.

At a high level, our argument drops the \emph{local orthogonalization} strategy of Barak et al.~\cite{BCK15} and instead runs the Gram--Schmidt procedure ``as-is''. Thus orthogonality of the resulting basis functions is immediate, and we need only show positive-semidefiniteness. We show that for any sequential ordering of the basis monomials in the Gram--Schmidt procedure, so long as it is of increasing degree, whenever we orthogonalize a monomial~$x^S$, the result basis function depends only on a small number of variables.

To see why such an assertion might be plausible, let us consider the task of orthogonalizing the singletons. The monomial basis may not orthogonal under $\pE[\cdot]$; e.g., consider the following $3$-XOR instance:
\begin{align*}
x_1 x_2 x_3 &= 1 & y_1 y_2 y_3 &= -1 \\
x_2 x_4 x_5 &= 1 & y_2 y_4 y_5 &= -1 \\
x_4 x_5 x_6 &= 1 & y_4 y_5 y_6 &= -1 \\
x_6 x_7 x_8 &= 1 & y_6 y_7 y_8 &= -1 \\
x_3 x_7 x_8 &= 1 & y_3 y_7 y_8 &= -1
\end{align*}
Observe that $x_1$ and $y_1$ each appear in exactly one constraint and all other variables each occur in exactly two constraints.  Multiplying each block of constraints together, we see that if $\pE[\cdot]$ satisfies all constraints then  $\pE[x_1] = 1$ and $\pE[y_1] = -1$.  So neither $x_1$ nor $y_1$ are orthogonal to $1$.  Since the two sets of equations are disjoint, we also know that $\pE[x_1 y_1] = -1$, so $x_1$ and $y_1$ are not orthogonal.  We note that many such blocks may occur in a random instance with $m \gg n^{1.4}$ constraints.  Let's try to understand  what happens when we run the Gram--Schmidt procedure on this basis. Consider an instance consisting of $n$ such disjoint blocks of $5$ constraints on $8n$ variables.  Let $x_{i1}$ be the variables that is fixed in block $i$.  Then every $x_{i1}$ is not orthogonal to $1$ and every pair $x_{i1}, x_{j1}$ is not orthogonal.  Intuitively, the variables $x_{i1}, x_{j1}$ behave independently, but are biased.  To fix this bias, consider the functions $\overline{x}_{i1}$ (where we use the notation $\overline{z} \coloneqq z - \pE[z]$).  Now we have that $\overline{x}_{i1}$ is orthogonal to $1$ and, by independence of the blocks, $\pE[\overline{x}_{i1} \cdot \overline{x}_{j1}] = 0$ for all $i,j$.


Ideally, we might hope this this new basis satisfies orthogonality when we move to degree~$2$, as well. Unfortunately, in general the basis $\{1,\overline{x}_1, \overline{x}_2, \ldots, \overline{x}_n, \overline{x_1 x_2}, \ldots\}$ again need not be orthogonal.  Consider a $3$-XOR instance with $n$ constraints $x_0 x_i y_i = b_i$ for $i \in [n]$; call this an $n$-star.  Random instances contain stars of superconstant size with high probability.  For all $\binom{n}{2}$ pairs $i,j$, it holds that $\overline{x_i y_i}$ and $\overline{x_j y_j}$ are not orthogonal under $\pE[\cdot]$:
\[
\pE[\overline{x_i y_i} \cdot \overline{x_j y_j}] = \pE[x_i y_i \cdot x_j y_j] - \pE[x_i y_i] \pE[x_j y_j] = b_i b_j - 0  = b_i b_j.
\]
Instead, consider the basis
\[
\wh{1} = 1, ~\wh{x_0} = x_0, ~\wh{x_1} = x_1, \,\ldots\,, ~\wh{y_1}= y_1, ~\wh{y_2} = y_2, \,\ldots\,, ~\wh{x_1 y_1} = x_1 y_1 - b_1 x_0, ~\wh{x_2 y_2} = x_2 y_2 - b_2 x_0, \,\ldots
\]
A simple calculation shows that these basis functions are orthogonal.  Each basis function depends on at most $3$ variables, so the degree-$3$ Sherali-Adams positivity condition and Fact~\ref{fact:gs-local} imply that degree-$2$ positive semidefiniteness holds.  We give a proof of orthogonality of $\wh{x_i y_i}$ and $\wh{x_j y_j}$ that illustrates the underlying intuition.  Observe that $\wh{x_i y_i}$ and $\wh{x_j y_j}$ are independent \emph{conditioned on $x_0$} for all $i \ne j$, and we can write
\begin{align*}
\pE[\wh{x_i y_i} \cdot \wh{x_j y_j}] &= \E[\wh{x_i y_i} \cdot \wh{x_j y_j}] \qquad &\text{($\pE[\cdot]$ is a valid expectation on small sets)} \\
&= \E[\E[\wh{x_i y_i} \cdot \wh{x_j y_j} | x_0]] \qquad &\text{(law of total expectation)} \\
&= \E[\E[\wh{x_i y_i} | x_0] \cdot \E[\wh{x_j y_j} | x_0]] \qquad &\text{(conditional independence of $\wh{x_i y_i}$ and $\wh{x_j y_j}$ given $x_0$)}.
\end{align*}
Next, note that
\[
\E[\wh{x_i y_i}| x_0=b] = \frac{1}{\Pr[x_0=b]}\E[\wh{x_i y_i} \cdot \indic{\{x_0 = b\}}(x_0)],
\]
where $\indic{\{x_0 = b\}}$ is the indicator function for $x_0 = b$.  Since we have orthogonalized $\wh{x_i y_i}$ against all degree-$1$ basis functions and $\indic{\{x_0 = b\}}$ is a degree-$1$ polynomial, this expression is equal to $0$.  Therefore, $\E[\wh{x_i y_i} | x_0] = 0$ and $\wh{x_i y_i}$ and $\wh{x_j y_j}$ are orthogonal.  In this case, $x_i y_i$ and $x_j y_j$ are correlated because they are connected by $x_0$.  After subtracting off their correlation with $x_0$, the resulting functions are orthogonal and no longer correlated.

Let us now formalize this intuition and generalize it to higher degree. At a high level, our idea is to show that the Gram--Schmidt process produces a basis such that each new basis element depends only on a small number of variables. Let $y_S$ be the result of applying the Gram--Schmidt process to $x^S$.  If $y_T$ appears in $y_S$ with a nonzero coefficient, then it must be the case that $\pE[x^S \cdot y_T] \ne 0$.  That is, $x^S$ and $y_T$ are correlated under $\pE[\cdot]$.  We show that this correlation is ``witnessed'' by some small, ``dense" subgraph containing many constraints covered by few variables.  If $y_S$ has many variables in its support, then there must be many such subgraphs.  We show that the union of these subgraphs is dense enough to be ``implausible''. This means that $y_S$ cannot have too many variables in its support.

Our witness can be seen as a generalization of the connected sets in the degree-$2$ case discussed above.  Call two sets of vertices $c$-connected if removing any set of $c-1$ vertices cannot disconnect them.  In the degree-$1$ case, nonzero correlation between $x^S$ and $y_T$ with $|S| = |T| = 1$ is witnessed by a small, dense, connected ($1$-connected) subgraph.  In the degree-$2$ case after orthogonalizing against degree-$1$ terms, we expect based on the star example that if $S$ and $T$ are only $1$-connected, then $x^S$ and $y_T$ will no longer be correlated.  We show that nonzero correlation between $x^S$ and~$y_T$ with $|S| = |T| = 2$ is then witnessed by a small, dense, $2$-connected subgraph.  In general, we show that nonzero correlation between $x^S$ and $y_T$ with $|S| = |T| = d$ is witnessed by a small, dense, $d$-connected subgraph.  This stronger connectivity requirement enables us to show that these witness subgraphs and their unions are dense enough to be implausible if the support of a basis function grows too large. For details of this argument, see Section~\ref{sec:psd-pf}.

\section{Forbidden subgraphs for the factor graph} \label{sec:plausible}

Let us make a few definitions concerning factor graphs, after which we will elaborate on the ``Plausibility Assumption''.
\begin{definition}[Subgraphs]
    We call $H$ a \emph{subgraph} of~$G$ if it is an \emph{edge-induced} subgraph; i.e., $H = G[A]$ for some subset~$A$ of the edges of~$G$.  We explicitly allow $A = \emptyset$ and hence $H = \emptyset$. The subgraph~$H$ need not be connected.
\end{definition}
\begin{notation}
    For $H$ a subgraph, we write $\vbls(H)$ for the set of variables appearing in~$H$, $\cons(H)$ for the set of constraints appearing in~$H$, and $\edges(H)$ for the set of edges appearing in~$H$.
\end{notation}
\begin{notation}
    Given $f \in \cons(H)$, we write $N_H(f) = \{i \in \vbls(H) : (f,i) \in \edges(H)\}$.  Note that this is \emph{not} necessarily the same thing as $N(f) \cap \vbls(H)$.
\end{notation}
We will typically measure the ``size'' of a subgraph by the number of constraints in it:
\begin{definition}[Small subgraphs]
    We say that subgraph $H$ is \emph{\smallish} if $|\cons(H)| \leq \CONSMALL$.
\end{definition}

Now regarding the Plausibility Assumption, for intuition's sake let us suppose we are concerned with weak refutation and degree-$O(1)$ SOS, as in Corollary~\ref{cor:noSDP}.  Thus we have some $k$-ary predicate~$P$ with $\cmplx(P) = \tau$, and we are selecting a random CSP with slightly fewer than $n^{\tau/2}$ constraints; say $m = n^{(\tau-\price)/2}$.  What does a random factor graph look like in this case? Which small subgraphs may appear?  A quick-and-dirty method to analyze this is as follows.  Consider the fixed  small subgraph in Figure~\ref{fig:1}; call it~$H$.
\myfig{.25}{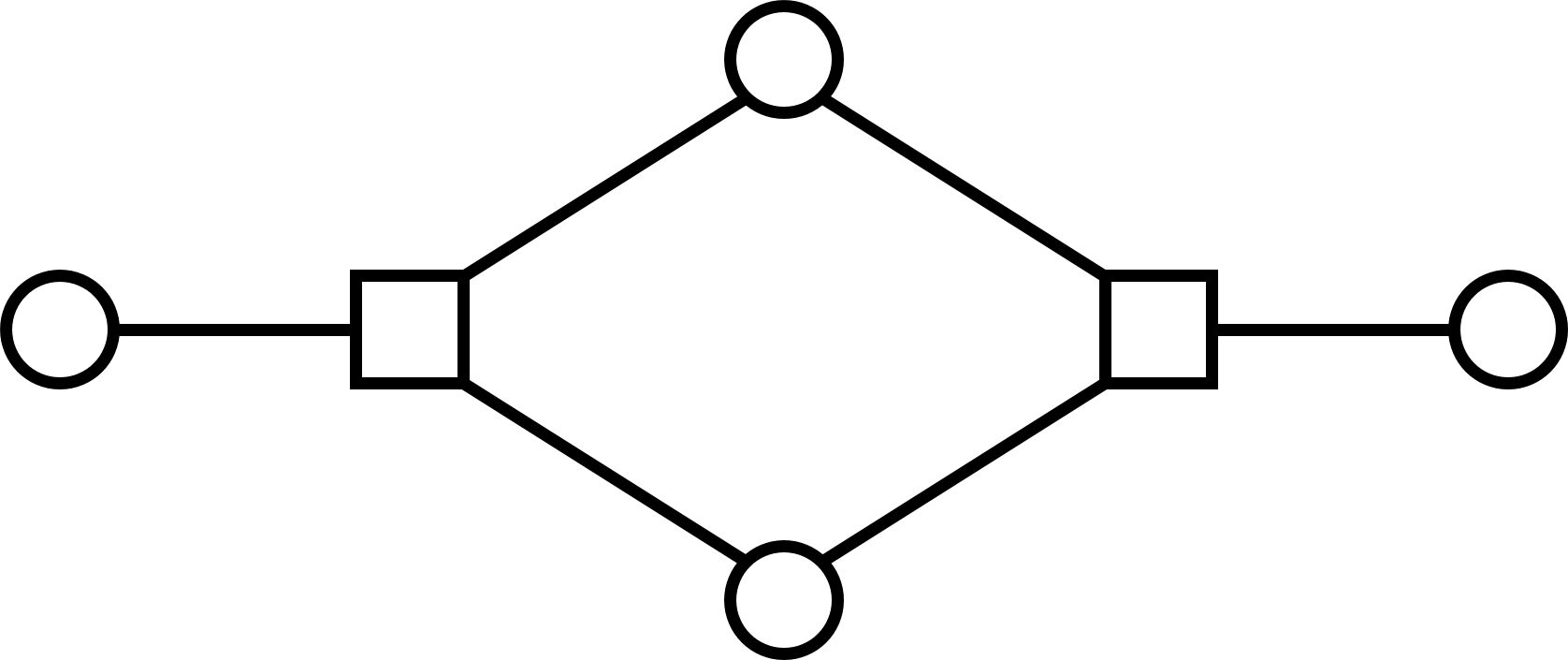}{An example small subgraph.  Constraint-vertices are squares, variable-vertices are circles.}{fig:1}
What is the expected number of copies of $H$ in a  random factor graph~$G$ with $n$ variable-vertices and $m = n^{(\tau-\price)/2}$ constraint-vertices?  There are $\binom{m}{2} \approx m^2$ choices for $H$'s $2$ constraint-vertices and $\binom{n}{4} \approx n^4$ choices for $H$'s $4$ variable-vertices.  Thinking of each constraint-vertex as choosing $k = O(1)$ random neighbors, the chance that the $6$ edges of~$H$  show up is roughly~$n^{-6}$.  Thus, very roughly, we expect about $m^{2} n^{4} n^{-6} = n^{2 \cdot (\tau - \price)/2 + (4-6)}$ copies of~$H$ in a random~$G$.  Thus copies of~$H$ ``plausibly'' show up if and only $2 \cdot (\tau - \price)/2 + (4-6) \geq 0$; i.e., if and only if $\tau \geq 2+\price$.  Since $\tau \geq 3$ always, this means we should certainly expect copies of~$H$ in~$G$.

For a general subgraph~$H$ with $c = |\cons(H)|$, $v = |\vbls(H)|$, $e = |\edges(H)|$,
\begin{equation} \label{eqn:plaus}
	\E[\#\text{ copies of } H] \approx m^{c}n^{v}n^{-e} = n^{c \cdot (\tau - \price)/2 + v - e} \ \implies\  H \text{ ``plausibly occurs'' iff } c \cdot (\tau - \price)/2 + (v - e) \geq 0.
\end{equation}
This inequality is precisely the one occurring in the Plausibility Assumption from Section~\ref{sec:plausible-fg}.

Despite the simple form of the inequality, we will find it helpful to view it in a different way.  For reasons that will become clear in Section~\ref{sec:PE}, we will be concerned almost exclusively with subgraphs of~$G$ in which all constraint-vertices have degree at least~$\Tcxty$:
\begin{definition}[$\tau$-subgraphs]
    Let $H$ be a subgraph.  We will call~$H$ a \emph{\subfactor} if  every constraint-vertex in~$H$ has degree at least~$\Tcxty$ within~$H$; i.e., $|N_H(f)| \geq \Tcxty$ for all $f \in \cons(H)$.
\end{definition}
\begin{remark}
    The empty subgraph~$\emptyset$ is always trivially a \subfactor.  Also, if $H$ and $H'$ are \subfactors then so is $H \cup H'$.
\end{remark}
\begin{definition}[Leaf vertices and interior vertices]
    Given a subgraph~$H$, we classify the variable-vertices in $H$ as either \emph{leaf} or \emph{interior} depending on whether they have degree $1$ or at least~$2$.  (Since $H$ is an edge-induced subgraph, it does not have any isolated vertices.)
\end{definition}

For \subfactors, there is a different way to view the ``plausibility inequality'' that will be more useful for us.  We define it with some ``accounting'' terminology.
\begin{definition}[Credit, debit, excess, revenue, cost, income]
    Let $H$ be a \subfactor. For the purposes of this definition, consider each of its edges to be two directed edges.
    \begin{itemize}
        \item For each variable-vertex, we assign it a \emph{credit} of~$1$ if it is a leaf vertex.  We'll write $\ell$ for the total credits.
        \item For each variable-vertex, any out-edges in excess of~$2$ are called \emph{excess}, and we assign a \emph{debit} for each.  We'll write $e_v$ for the total number of these.
        \item For each constraint-vertex, any out-edges in excess of~$\Tcxty$ are called \emph{excess}, and we assign a debit for each.  We'll write $e_c$ for the total number of these, and $e = e_c + e_v$ for the total debit (number of excess edges).
        \item The sum of credits minus the sum of debits, $\ell - e$, is called the \emph{revenue}.  We denote it by~$R(H)$.
        \item Each constraint-vertex has a \emph{cost} of $\price$. We write $C(H) = \price \cdot \Size{H}$ for the total \emph{cost}.
        \item The \emph{income} is $I(H) = R(H) - C(H)$.
    \end{itemize}
\end{definition}
\begin{definition}[Plausible \subfactors]
    Let $H$ be a \subfactor.  We say that $H$  is \emph{plausible} if~$I(H) \geq 0$.
\end{definition}
\begin{remark}
	$H$ being plausible implies (indeed, is equivalent to) $\Size{H} \leq \frac{1}{\price} \cdot R(H)$. Thus controlling a subgraph's revenue is equivalent to controlling its size.
\end{remark}

The next lemma implies that the inequality $I(H) \geq 0$ is the same as the inequality appearing in the Plausibility Assumption and in~\eqref{eqn:plaus}.
\begin{lemma}                                     \label{lem:alternate-accounting}
    Let $H$ be a \subfactor with $c = \Size{H}$, $v = |\vbls(H)|$, $e = |\edges(H)|$, and $I = I(H)$.  Then $e = \frac{\Tcxty-\price}{2}\cdot c + v - \frac{I}{2}$.
\end{lemma}
\begin{proof}
    We count the number of ``directed edges'' in~$H$. Counting those coming out of variable-vertices, the $\ell$ leaf vertices contribute~$1$ each, and the $v-\ell$ interior vertices contribute $2(v-\ell) + e_v$. Counting the directed edges coming out of constraint-vertices yields $\Tcxty c + e_c$.  Thus
    \[
        \text{\#\ directed edges} = 2e = \ell + 2(v-\ell) + e_v + \Tcxty c + e_c = \Tcxty c + 2v - (\ell -  e) = \Tcxty c + 2v - (\price c+I),
    \]
    since $\ell - e = R(H) = C(H) + I(H)$.  The claim follows.
\end{proof}

In light of this, we may restate the Plausibility Assumption:
\paragraph{Plausibility Assumption, Restated.} \emph{Henceforth we assume the factor graph~$G$ has the following property: All \subfactors~$H$ of~$G$ with $\Size{H} \leq 2 \cdot \CONSMALL$ are plausible.}\\

As mentioned earlier, for an appropriate choice of $\CONSMALL$, the Plausibility Assumption holds for a random instance. More precisely, in  Appendix~\ref{app:expansion} we prove the below theorem.  The reader is advised that in this theorem, the first claim is the main one; it is used to show our Theorem~\ref{thm:main2} concerning weak refutation.  The second claim (``Moreover\dots'') is a technical variant needed to extend our results to give Theorem~\ref{thm:main-intro} concerning $\delta$-refutation.
\begin{theorem}     \label{thm:random-graph}
    Let $\Tone = \Tcxty - 2 \geq 1$.   Fix $0 < \price \leq .99 \Tone$,  $0 < \beta < \frac12$. Then except with probability at most $\beta$, when $\bG$~is a random instance with $m = \density n$ constraints, the Plausibility Assumption holds provided
    \[
        \CONSMALL \leq \gamma \cdot \frac{n}{\density^{2/(\Tone-\price)}},
    \]
    where $\gamma = \frac{1}{\maxarity} \left(\frac{\beta^{1/\Tone}}{2^{\maxarity/\Tone}}\right)^{O(1)}$.   Moreover, assuming $\price < 1$, except with probability at most~$\beta$ we have
	\[
	\#\{\textnormal{nonempty \subfactors~$H$ with $\cons(H) \leq 2 \cdot \CONSMALL$} : I(H) \leq \Tcxty - 1\} \leq \density n^{\frac{1+\price}{2}}.
    \]
\end{theorem}

\section{Defining the pseudoexpectation} \label{sec:PE}
\subsection{Closures}
In this section we define the ``closure'' of a set of variables.  Roughly speaking, this can be thought of as the smallest \subfactor of~$G$ that fully determines the distribution on~$S$ under a natural ``planted distribution''.

\begin{definition}[$S$-closed subgraph]
    Let $S$ be a set of variables.  We say that a subgraph $H$ is \emph{$S$-closed} if it is a \subfactor and all its leaf vertices are in~$S$.
\end{definition}
\begin{remark}      \label{rem:constr-in-cl}
    For every constraint in $G$, if $H$ is taken to be the full neighborhood of that constraint, and $S$ is the set of variables in that constraint, then $H$ is $S$-closed.
\end{remark}
Note that a union of $S$-closed \subfactors is $S$-closed. This leads us to the following definition:
\begin{definition}[Closure, $\cl(S)$]
    Let $S$ be a set of variables. We define the \emph{closure} of $S$, written $\cl(S)$, to be the union of all \emph{\smallish} $S$-closed \subfactors~$H$.  Note that  $\cl(S)$ is itself an $S$-closed \subfactor.
\end{definition}
\begin{remark}
    A key warning to remember: we do not necessarily have $S \subseteq \vbls(\cl(S))$.
\end{remark}
\begin{remark}  \label{rem:ST-closure}
    Let $T \subseteq S$.  Then if $H$ is $T$-closed, it is also $S$-closed.   It follows that $\cl(T) \subseteq \cl(S)$.
\end{remark}

\begin{fact}        \label{fact:cl-emptyset}
    The only plausible $\emptyset$-closed \subfactor $H$ is $H = \emptyset$.  It follows that $\cl(\emptyset) = \emptyset$.
\end{fact}
\begin{proof}
    If $H$ is $\emptyset$-closed then its revenue is at most~$0$.  Hence if it is plausible, its cost is~$0$.
\end{proof}

We will now give an important generalization of this fact for $S$-closures, $|S| > 0$
\begin{theorem}                                     \label{thm:closures-are-small}
    Let $S$ be a set of variables with $|S| \leq \price \cdot \CONSMALL$.  Then $\cl(S)$ is \smallish and satisfies $R(\cl(S)) \leq |S|$.
\end{theorem}
\begin{proof}
    Since $\cl(S)$ is $S$-closed, all its leaf vertices are in~$S$; thus $\cl(S)$ has at most $|S|$ credits and so $R(\cl(S)) \leq |S|$, as claimed.  Observe that if $H_1, \dots, H_t$ is the complete list of $S$-closed \subfactors, we may make the same deduction about $H_1 \cup \cdots \cup H_j$ for any $1 \leq j \leq t$, in particular deducing that $R(H_1 \cup \cdots \cup H_j) \leq \price \cdot \CONSMALL$ for each~$j$.  The \smallness of $\cl(S)$ is now a consequence of the lemma that immediately follows.
\end{proof}
\begin{lemma}                                       \label{lem:smallness-trick}
    Suppose that $H$ is a \subfactor formed as a union, $H = H_1 \cup \cdots \cup H_t$, where each $H_j$ is \smallish and where we have $R(H_1 \cup \cdots \cup H_j) \leq \price \cdot \CONSMALL$ for all $1 \leq j \leq t$.  Then $H$ is \smallish.
\end{lemma}
\begin{proof}
    The proof is by induction on~$t$, with the base case of $t = 1$ being immediate.  In general, suppose $H' = H_1 \cup \cdots \cup H_{t-1}$ is \smallish.  Since $H_t$ is also \smallish we have $|\cons(H')|, |\cons(H_t)| \leq \CONSMALL$ and hence $|\cons(H' \cup H_t)| \leq 2 \cdot \CONSMALL$.  Thus $H' \cup H_t$ is plausible and so
    \[
        \cons(H' \cup H_t) = (1/\price) \cdot C(H' \cup H_t) \leq (1/\price) \cdot R(H' \cup H_t) \leq (1/\price) \cdot \price \cdot \CONSMALL = \CONSMALL,
    \]
    showing that $H' \cup H_t$ is \smallish, completing the induction.
\end{proof}

In proving Theorem~\ref{thm:closures-are-small}, we iteratively formed the union of all \smallish $S$-closed subgraphs, at each step verifying that we have a \smallish \subfactor of revenue at most~$|S|$.  Once we finish producing $\cl(S)$ in this way, let $V = \vbls(\cl(S))$, and suppose we \emph{continue} iteratively adding in \smallish \subfactors that are $(S \cup V)$-closed.  This process cannot add any leaf vertices except possibly in~$S$; thus we will still have that revenue is bounded by~$|S| \leq \price \cdot \CONSMALL$,  and Lemma~\ref{lem:smallness-trick} will still imply the resulting \subfactor is \smallish.  Thus we end up with a \smallish, $S$-closed \subfactor --- which by definition is already contained in~$\cl(S)$.  Thus we have shown:
\begin{theorem} \label{thm:iterated-closure-is-closure}
    Let $S$ be a set of variables with $|S| \leq \price \cdot \CONSMALL$.  Then $\cl(S \cup \vbls(\cl(S))) = \cl(S)$.
\end{theorem}

\subsection{The planted distribution}

\begin{definition}[Planted distribution on a small subgraph]
    Let $H$ be a \smallish subgraph of~$G$. The \emph{planted distribution on~$H$} is a probability distribution on assignments~$\bx \in \Omega^n$ to the variables of~$G$, defined as follows:  For each constraint $f \in \cons(H)$ we independently draw an assignment $\bw_f \in \Omega^{N(f)}$ according to~$\mu_f$.  We write its component associated to variable $i \in N_H(f)$ as $\bw_{f,i}$, and think of it as an assignment ``suggested'' for this variable.  (Note that we will ignore the components of $\bw_f$ correspoding to variables not in~$N_H(f)$.)  Now each variable $i \in \vbls(H)$ has one or more assignments in~$\Omega$ suggested by its adjacent constraints. We get a unique assignment~$\bx_i$ for it by \emph{conditioning} on all the suggestions being consistent.  (We will show later in~\eqref{eqn:consist-nonzero} that this occurs with nonzero probability.) Finally, assignments for variables not in~$H$ are chosen independently and uniformly from~$\Omega$.

    We'll write $\eta_H$ for the probability distribution on $\Omega^n$ associated to this planted distribution on~$H$, and we'll write $\E_H[\cdot]$ for the associated expectation.
\end{definition}

\begin{definition}
    For each $i \in \vbls(G)$ and each $c \in \Omega$, we introduce an ``indeterminate'' $\indet{c}{i}$ that is supposed to stand for~$1$ if variable~$i$ is assigned~$c$ and~$0$ otherwise.
\end{definition}

The key theorem about the planted distributions is that as soon as a subgraph~$H$ contains~$\cl(S)$, the marginal of~$\eta_H$ on~$S$ is determined.  In some sense, this property is exactly the reason we defined the closure the way we did.

\begin{theorem}                                     \label{thm:closure}
    Let $S$ be a set of variables and let $H \supseteq \cl(S)$ be a \smallish subgraph.
    Then the marginal of $\eta_H$ on $S$ is the same as the marginal of $\eta_{\cl(S)}$ on~$S$.
\end{theorem}
\begin{remark}
    Although the notation in the below proof looks cumbersome, the calculations are actually fairly straightforward.  We strongly encourage the reader to work through the proof in the case of $q = 2$, $\Omega = \{\pm 1\}$, with ``$\olindet{c}{i}$'' replaced by $cx_i \in \{\pm 1\}$.
\end{remark}
\begin{proof}
    For brevity we write $v_H = |\vbls(H)|$ and $e_H = |\edges(H)|$.
    We also introduce the notation $\olindet{c}{i} = q\indet{c}{i} - 1$.  Recalling that $\eta_H$ puts the uniform distribution on the $n-v_H$ variables outside~$H$, we have
    \begin{align}
        p_H(x) &\coloneqq \Pr_{\bw}[\bw \text{ suggestions consistent, and the assignment to $\vbls(H)$ they agree on is~$x$}] \nonumber\\
            &= q^{-(n-v_H)} \cdot \E_{\bw} \prod_{(f,i) \in \edges(H)} \left(q^{-1}(1+\olindet{\bw_{f,i}}{i}\right)\nonumber\\
            &= q^{v_H - e_H - n} \cdot  \sum_{H' \subseteq H} \E_{\bw} \prod_{(f,i) \in \edges(H')} \olindet{\bw_{f,i}}{i} \label{eqn:deduce0}\\
            &= q^{v_H - e_H - n} \cdot \sum_{H' \subseteq H} \prod_{f \in \cons(H')} \E_{\bw_f \sim \mu_f} \prod_{i \in N_{H'}(f)} \olindet{\bw_{f,i}}{i}, \label{eqn:deduce-from-me}
    \end{align}
    where we used that the draws $\bw_f \sim \mu_f$ are independent across $f$'s.  Now whenever $f \in \cons(H')$ has $|N_H(f)| < \Tcxty$, the $(\Tcxty-1)$-wise uniformity of~$\mu_f$ implies that
    \[
        \E_{\bw_f \sim \mu_f} \prod_{i \in N_{H'}(f)} \olindet{\bw_{f,i}}{i} = \E_{\substack{\bw_{f,i} \sim \Omega \\ \text{uniform, indep.}}} \prod_{i \in N_{H'}(f)} \olindet{\bw_{f,i}}{i} =  \prod_{i \in N_{H'}(f)} \E_{\substack{\bw_{f,i} \sim \Omega \\ \text{uniform}}}\olindet{\bw_{f,i}}{i} = 0,
    \]
    since $\E_{\bc \sim \Omega}[\olindet{\bc}{i}] = 0$ for any fixed value~$\x \in \Omega$.  Thus in~\eqref{eqn:deduce-from-me} it is equivalent to sum over \subfactors~$H'$, and so returning to~\eqref{eqn:deduce0} we get
    \begin{equation} \label{eqn:use-me}
        p_H(x) 
        = q^{v_H - e_H - n} \cdot \sum_{\substack{H' \subseteq H \\ \mathclap{H' \text{ a \subfactor}}}} \ \ \ \ \E_{\bw} \prod_{(f,i) \in \edges(H')} \olindet{\bw_{f,i}}{i}.
    \end{equation}
    Suppose now that $T \subseteq [n]$ is a set of variables.  We'll decompose an $x \in \Omega^n$ into its projection $x_T$ onto the coordinates in~$T$ and $x_{\ol{T}}$ onto the coordinates not in~$T$.  Then
    \begin{align}
        p_H(x_T) &\coloneqq \Pr_{\bw}[\bw \text{ suggestions consistent, and the assignment to $T$ they agree on is~$x_T$]} \nonumber\\
        &= \sum_{x_{\ol{T}} \in \Omega^{\ol{T}}} p_H(x_T,x_{\ol{T}}) = q^{n-|T|} \cdot \E_{\substack{\bx_{\ol{T}} \sim \Omega^{\ol{T}} \\ \text{uniform}}} [p_H(x_T,\bx_{\ol{T}})] \nonumber\\
        &= q^{v_H - e_H - |T|} \cdot \sum_{\substack{H' \subseteq H \\ \mathclap{H' \text{ a \subfactor}}}} \quad \E_{\bw} \prod_{\substack{(f,i) \in \edges(H') \\ i \in T}} \olindet{\bw_{f,i}}{i}  \cdot  \E_{\substack{\bx_{\ol{T}} \sim \Omega^{\ol{T}} \\ \text{uniform}}} \prod_{\substack{(f,i) \in \edges(H') \\ i \in \ol{T}}} \olbindet{\bw_{f,i}}{i}, \label{eqn:nonoggin}
    \end{align}
    where we used~\eqref{eqn:use-me}. Now suppose the \subfactor $H'$ has a leaf vertex~$j$ that is in~$\ol{T}$; i.e., it's \emph{not} in~$T$.  Then $\bx_j$ appears exactly once in the above, within the expression
    \begin{equation}    \label{eqn:will-be-0}
         \E_{\substack{\bx_{\ol{T}} \sim \Omega^{\ol{T}} \\ \text{uniform}}} \prod_{\substack{(f,i) \in \edges(H') \\ i \in \ol{T}}} \olbindet{\bw_{f,i}}{i}.
    \end{equation}
    As $\bx_j$ is chosen uniformly and independently of all random variables, the above contains a factor of the form $\E_{\bx_j \sim \Omega}[\olbindet{\bw_{f,j}}{j}]$. But for \emph{any} fixed outcome of $\bw_{f,j}$, this expectation is~$0$, meaning~\eqref{eqn:will-be-0} will vanish.  Thus any summand $H'$ in~\eqref{eqn:nonoggin} will vanish if $H'$ has a leaf variable outside~$T$.  Thus we may equivalently sum only over \emph{$T$-closed}~$H'$.  That is,
    \begin{align}
        p_H(x_T) &= \Pr_{\bw}[\bw \text{ suggestions consistent, and the assignment to $T$ they agree on is~$x_T$]} \nonumber\\
        &= q^{v_H - e_H - |T|} \cdot \sum_{\substack{H' \subseteq H \\ \mathclap{H' \text{ is $T$-closed}}}} \quad \E_{\bw} \prod_{\substack{(f,i) \in \edges(H') \\ i \in T}} \olindet{\bw_{f,i}}{i}  \cdot  \E_{\substack{\bx_{\ol{T}} \sim \Omega^{\ol{T}} \\ \text{uniform}}} \prod_{\substack{(f,i) \in \edges(H') \\ i \in \ol{T}}} \olbindet{\bw_{f,i}}{i} \nonumber\\
        &= q^{v_H - e_H - |T|} \cdot \sum_{\substack{H' \subseteq H \\ \mathclap{H' \text{ is $T$-closed}}}} \quad \E_{\bw} \quad \E_{\substack{\bx \sim \Omega^n \text{ unif.,} \\ \text{condit.\ on }\bx_T = x_T}} \ \prod_{(f,i) \in \edges(H')} \olbindet{\bw_{f,i}}{i}. \label{eqn:planted-formula}
    \end{align}
    Suppose we took $T = \emptyset$ above.  Since $H$ is \smallish, every subgraph $H'$ is plausible and hence Fact~\ref{fact:cl-emptyset} implies that the above has only one summand, corresponding to~$H' = \emptyset$.  The summand is trivially~$1$, and hence
    \begin{equation}    \label{eqn:consist-nonzero}
        \Pr_{\bw}[\bw \text{ suggestions consistent}] = q^{v_H - e_H}.
    \end{equation}
    Observe that this does not depend at all on the $\mu_f$'s; in particular, it is easily seen to the be the probability of consistent suggestions under completely uniform~$\mu_f$'s.  In any case, since~\eqref{eqn:consist-nonzero} is positive, as promised, we may condition on the associated event; thus from~\eqref{eqn:planted-formula} we obtain
    \begin{multline*}
        \Pr_{\bw}[\text{the suggested assignment to $S$ is~$x_S$} \mid \text{the  suggestions $\bw$ are consistent}] \\
        = q^{- |S|} \cdot \sum_{\substack{H' \subseteq H \\ \mathclap{H' \text{ is $S$-closed}}}} \quad \E_{\bw} \quad \E_{\substack{\bx \sim \Omega^n \text{ unif.,} \\ \text{condit.\ on }\bx_S = x_S}} \ \prod_{(f,i) \in \edges(H')} \olbindet{\bw_{f,i}}{i}.
    \end{multline*}
    This formula visibly has the property that once $H \supseteq \cl(S)$, it does not depend on~$H$.
\end{proof}

\subsection{Pseudoexpectations} \label{sec:pE}
In this section, we formally define the pseudoexpectation with which we will work.
\begin{definition}
Given a polynomial expression $p(\x)$ in the indeterminates $\indet{c}{i}$, we write
    \begin{align*}
        \vblsp(p) &= \{i : \text{at least one $\indet{c}{i}$ appears in~$p(\x)$}\}, \\
        \degmlin(p) &= \max\{|\vblsp(M)| : \text{$M(\x)$ is a monomial in $p(\x)$}\}.
    \end{align*}
    We call the latter the \emph{\multilineardegree}; note that $\degmlin(p) \leq \deg(p)$ always.

    Recall that a \emph{pseudoexpectation} on polynomials of degree at most~$D$ is a linear map $\pE[\cdot]$ satisfying~$\pE[1] = 1$.  We can uniquely define it by specifying its values on all monomials of degree at most~$D$.  Further, recall that if $p(\x)$ is a polynomial, 
     we say that $\pE[\cdot]$ \emph{satisfies the identity $p(\x) = 0$} if $\pE[p(\x)\cdot q(\x)] = 0$ for all polynomials $q(\x)$ with $\deg(p\cdot q) \leq D$.
\end{definition}

\begin{definition}[Our pseudoexpectation]
    We'll define our pseudoexpectation $\pE[\cdot]$ on all polynomials of \emph{\multilineardegree} at most~${\price \cdot \CONSMALL}$; in particular, this defines it for all polynomials of (usual) degree at most~${\price \cdot \CONSMALL}$.  We define it by imposing that $\pE[M(\x)] = \E_{\cl(\vblsp(M))}[M(\bx)]$ for all monomials $M(\x)$ having $\degmlin(M) \leq \price \cdot \CONSMALL$.  (Here we are using the abbreviation $\E_C[M(\bx)]$ for $\E_{\bx \sim \eta_C}[q(\bx)]$.)  By Theorem~\ref{thm:closures-are-small}, this makes sense in that $\cl(\vbls(M))$ will always be \smallish.  Note that we have $\pE[1] = \E_{\cl(\emptyset)}[1] = 1$, as required.
\end{definition}
\begin{theorem}                                     \label{thm:consistent-distribution}
    Let $p(x)$ be a polynomial expression of \multilineardegree at most~$\price \cdot \CONSMALL$.  Let $H$ be any \smallish subgraph containing
    \[
        H' = \bigcup \{ \cl(\vblsp(M)) : \text{$M(\x)$ is a monomial of $p(\x)$}\}.
    \]
    For example, if $\cl(\vblsp(p))$ is \smallish then it would qualify for~$H$.  Then
    \[
        \pE[p(\x)] = \E_{H}[p(\bx)] = \E_{H'}[p(\bx)].
    \]
\end{theorem}
\begin{proof}
    This is immediate from Theorem~\ref{thm:closure} and Remark~\ref{rem:ST-closure}.
\end{proof}

\begin{theorem}         \label{thm:pE-sat}
    Let $p(\x)$ be a polynomial  with $S = \vblsp(p)$ satisfying $|S| \leq \deg(p)$, $|S| \leq \price \cdot \CONSMALL$. Assume that $p(\bx)$ is identically zero for $\bx \sim \eta_{\cl(S)}$.  (Note that $\cl(S)$ is \smallish by Theorem~\ref{thm:closures-are-small}.) Then our $\pE[\cdot]$ satisfies the identity $p(\x) = 0$.
\end{theorem}
\begin{proof}
    Let $q(\x)$ be a nonzero polynomial with $\deg(p \cdot q) \leq \price \cdot \CONSMALL$.   Writing $q(\x) = \sum_j M_j(\x)$ where each $M_j(\x)$ is a monomial, we have
    \begin{equation}    \label{eqn:i-should-be-0}
        \pE[p(\x) \cdot  q(\x)] = \sum_j \pE[p(\x) \cdot  M_j(\x)] = \sum_j \E_{\cl(S \cup \vblsp(M_j))}[p(\bx) \cdot  M_j(\bx)].
    \end{equation}
    Here the last equality used Theorem~\ref{thm:consistent-distribution} and the \smallness of $\cl(S \cup \vblsp(M_j))$, which follows from Theorem~\ref{thm:closures-are-small} and the fact that $|S \cup \vblsp(M_j)| \leq \deg(p) + \deg(q)  = \deg(p \cdot q) \leq \price \cdot \CONSMALL$.  But since $\cl(S\cup \vblsp(M_j)) \supseteq \cl(S)$ (Remark~\ref{rem:ST-closure}), Theorem~\ref{thm:closure} tells us that $p(\bx)$ has the same distribution under $\eta_{\cl(S\cup \vblsp(M_j))}$ and $\eta_{\cl(S)}$; i.e., it is identically~$0$.  Thus~\eqref{eqn:i-should-be-0} vanishes, as needed.
\end{proof}
We have the following immediate corollaries:
\begin{corollary}                                       \label{cor:pE-sats-stuff}
    Our pseudoexpectation $\pE[\cdot]$ satisfies the following identities:
    \begin{itemize}
        \item $\sum_{c \in \Omega} \indet{c}{i} = 1$ for all $i \in [n]$ (i.e., the identity $\sum_{c \in \Omega} \indet{c}{i} - 1 = 0$).
        \item $\indet{c}{i}^2 = \indet{c}{i}$ for all $c \in \Omega, i \in [n]$.
    \end{itemize}
    As an immediate consequence of the latter,  we always have $\pE[p(\x)] = \pE[\mathrm{multilin}(p(\x))]$, where $\mathrm{multilin}(p(\x))$ is defined by replacing any positive power of $\indet{c}{i}$ in~$p(\x)$ with just $\indet{c}{i}$.
\end{corollary}
Another corollary is the following (cf.~the rough statement of our main technical result, Theorem~\ref{thm:technical-approx}):
\begin{corollary}   \label{cor:pE-sats-constraints}
    Our pseudoexpectation $\pE[\cdot]$ satisfies the identity
    \[
        s_f(\x) \coloneqq \sum_{\vec{c} \in \supp(\mu_f)} \prod_{i \in N(f)} \indet{c_i}{i} = 1
    \]
    for all $f \in \cons(G)$; i.e., ``$\pE[\cdot]$'s~distribution on $N(f)$ is always in $\supp(\mu_f)$''.
\end{corollary}
\begin{proof}
    We apply Theorem~\ref{thm:pE-sat}, with $S = N(f)$, which satisfies $|S| =\deg(s_f)$ and  $|S| \leq \maxarity \leq \price \cdot \CONSMALL$. Note that if $H_f$ denotes the \subfactor induced by all edges of~$G$ incident on constraint-vertex~$f$, then $H_f$ is $S$-closed and so $H_f \subseteq \cl(S)$.  It then follows from the definition of $\bx \sim \eta_{\cl(S)}$ that $s_f(\bx) \equiv 1$, since the restriction of $\bx$ to $N(f)$ will always be supported on $\supp(\mu_f)$.
\end{proof}

\section{The proof of positive semidefiniteness} \label{sec:psd-pf}
\subsection{Setup}
Throughout this section, fix a degree~$D$ satisfying $1 \leq D \leq \frac13 \price \cdot \CONSMALL$. \todo{1/3 would be 1/2 if we could get r not r+d in key lemma assumption}Our goal will be to establish:
\begin{theorem}                                     \label{thm:main}
    If $p(x)$ is a polynomial expression of degree at most~$D$, then $\pE[p(x)^2] \geq 0$.
\end{theorem}
In light of Corollary~\ref{cor:pE-sats-stuff}, we may assume that $p(x)$ is ``multilinear'' (i.e., does not contain~$\indet{c}{i}^k$ for any $k > 1$).  Another way to state this assumption is $p(x) \in \spn(x^S : S \in \calM^{\leq D})$, where we introduce the following notation:
\begin{definition}
    A \emph{monomial index} will be a set~$S$ of pairs $(i,c) \in [n] \times \Omega$, with no variable $i \in [n]$ occurring more than once.  We write $x^S$ for the monomial $\prod_{(i,c) \in S} \indet{c}{i}$, with the usual convention that $x^\emptyset = 1$.   Finally, we write $\calM^{\leq D}$ for the collection of monomial indices~$S$ with $|S| \leq D$.
\end{definition}
\begin{notation}    \label{not:abuse}
    We abuse notation as follows: If a monomial index $S$ occurs in a place where a subset of variables is expected, we intend the subset of variables $\{i : (i,c) \in S \text{ for some $c$}\}$.
\end{notation}
\begin{remark}
    All of the ideas in our proof of Theorem~\ref{thm:main} are present in the $q = 2$ case; only notational complexities arise for $q > 2$.  Thus the reader is encouraged to keep the Boolean case $\Omega = \{\text{false}, \text{true}\}$ in mind.  In this case, since $\pE[\cdot]$ satisfies the identity $\indet{\text{false}}{i} = 1 - \indet{\text{true}}{i}$, one can also ignore the indeterminate $\indet{\text{false}}{i}$ (since $1 \in \spn(\calM^{\leq 0})$ already).  Then one can more naturally write the indeterminate $\indet{\text{true}}{i}$ as $x_i$ and the monomial $x^S$ becomes $\prod_{i \in S} x_i$.
\end{remark}

\subsection{Gram--Schmidt overview}
\begin{notation}
    Let $\preceq$ denote any total ordering on $\calM^{\leq D}$ that respects cardinality, so that if $T$ and $S$ are monomial indices with $|T| < |S|$, then $T \prec S$.  For $S \neq \emptyset$, let $\prev(S)$ denote the immediate predecessor of~$S$ under~$\preceq$.
\end{notation}

Our goal in this section is to show that the \emph{modified Gram--Schmidt process} from linear algebra can be successfully applied to the monomials $({x}^S : S \in \calM^{\leq D})$, in the ordering~$\preceq$, using $\pE[\cdot]$ as the ``inner product'': $\la p(x), q(x) \ra \coloneqq \pE[p(x)\cdot q(x)]$.  
Of course, we don't know that this is a genuine inner product (indeed, that's essentially what we're trying to prove).  We will discuss this issue shortly, but we first remind the reader that the modified Gram--Schmidt process would typically produce a collection of polynomials $y_S = y_S(x)$, for $S \in \calM^{\leq D}$, that are \emph{orthogonal} under~$\pE[\cdot]$ (meaning $\pE[y_{S} \cdot y_{S'}] = 0$ if $S \neq S'$) and that have the same span as $({x}^S : S \in \calM^{\leq D})$.  As well, it would produce ``normalized'' versions of these polynomials~$z_S = y_{S}/\sqrt{\pE[y_S^2]}$, satisfying $\pE[z_S^2] = 1$.

We now address the obviously difficulty that $\pE[\cdot]$ is not (known to be) an inner product, because we don't know it's positive definite on the monomials of~$\calM^{\leq D}$. Our goal will be to show that as we follow the Gram--Schmidt process, it never encounters any ``positive definiteness problems'', and therefore ``succeeds''.  The main ``positive definiteness problem'' Gram--Schmidt might encounter would be if it creates a polynomial~$y_S$ with $\pE[y_S^2] < 0$. In this case, when it tries to produce the normalized polynomial~$z_S$, it would certainly fail.

There is one additional potential problem, occurring if Gram--Schmidt produces a~$y_S$ with $\pE[y_S^2] = 0$.  In the usual process from linear algebra this may indeed occur, and the Gram--Schmidt algorithm copes by treating $z_S$ as~$0$ (effectively, throwing it out of the span).  This is a valid strategy because genuine inner products are strictly positive definite.  However we only expect our ``inner product'' $\pE[\cdot]$ to be positive \emph{semi}definite.  We therefore need a different coping mechanism.  For us, when $\pE[y_S^2] = 0$ occurs, we will simply define its ``normalized'' version~$z_S$ to be~$y_S$. The challenge of this is that Gram--Schmidt's guarantee of producing an orthogonal collection $(y_S : S \in \calM^{\leq D})$ relies syntactically on all the~$z_S$ polynomials satisfying $\pE[z_S^2] = 1$.  Thus we will have an additional burden: we will have to ``manually'' show that $\pE[y_S^2] = 0$ implies that $y_S$ is orthogonal under $\pE[\cdot]$ to all other polynomials. It will count as a ``positive definiteness problem'' if we are unable to show this; we will call this the ``pseudovariance zero problem''.  We remark that the main positive definiteness problem is  fundamentally more important than this ``pseudovariance zero problem'', and the reader may wish to ignore the pseudovariance zero issue on first reading.

We now describe the modified Gram--Schmidt process in detail.  The process works in \emph{stages}, named after the elements of~$\calM^{\leq D}$ and in order of~$\preceq$.  At the end of stage~$S$ it creates a certain polynomial~$z_S$.  Stage~$\emptyset$ always ``succeeds'' and simply consists of defining $z_\emptyset = 1$.  In some cases it may happen that $\pE[z_S^2] = 0$.  In this case we say that $z_S$ has \emph{pseudovariance zero}, and the Gram--Schmidt algorithm will add~$S$ to a growing collection called~$\PvZ$.

Each stage~$S$ is further divided into \emph{substages}, associated to monomial indices $T \prec S$ in order of~$\preceq$.  Let us introduce some notation:
\begin{notation}
    Let $\calM^{\leq D}_2$ denote the collection of all pairs $(S,T) \in \calM^{\leq D} \times \calM^{\leq D}$ with $T \prec S$.  We define a total ordering $\preceq_2$ on~$\calM^{\leq D}_2$ via
    \[
        (S',T') \preceq_2 (S,T) \iff S' \prec S, \text{ or } S' = S \text{ and } T' \preceq T.
    \]
\end{notation}
Thus the overall progression of substages in Gram--Schmidt is through the elements of~$\calM^{\leq D}_2$ in order of~$\preceq_2$.  Substage~$(S,T)$ creates a polynomial~$y_{S,T}$ as follows:
\[
    y_{S,T} = \begin{cases}
                        {x}^S - \pE[{x}^S] & \text{if $T = \emptyset$;} \\
                        y_{S, \prev(T)} - \pE[y_{S,\prev(T)} \cdot z_{T}] z_{T} & \text{else.}
                    \end{cases}
\]
Stage $S$ ends just after substage $(S,\prev(S))$.  At this point, the Gram--Schmidt process defines
\[
    y_S = y_{S, \prev(S)}, \qquad
    z_{S} = \begin{cases}
                        y_{S}\bigm/\sqrt{\pE[y_{S}^2]} & \text{if $\pE[y_{S}^2] > 0$;} \\
                        y_{S} & \text{if $\pE[y_{S}^2] = 0$, in which case $S$ is placed into $\PvZ$.}
                  \end{cases}
\]
Of course, if $\pE[y_S^2] < 0$ then we have encountered a positive definiteness problem.  Indeed, to be conservative we will treat it as a problem if $\pE[y_{S,T}^2] < 0$ for \emph{any} $(S,T) \in \calM^{\leq D}_2$.

It is a syntactic property of the usual modified Gram--Schmidt process that when~$y_{S,T}$ is produced, it is orthogonal to~$z_T$ under $\pE[\cdot]$.  However this relies on $\pE[z_T^2] = 1$, which fails for us if $T \in \PvZ$.  Thus we will need to explicitly prove that $T \in \PvZ$ implies $\pE[y_{S,\prev(T)} \cdot z_T] = 0$.  If this doesn't hold, we've encountered the pseudovariance zero problem.  But assuming it does hold, $y_{S,T}$ will simply become $y_{S,\prev(T)}$ and we will have the desired orthogonality of $y_{S,T}$ and~$z_T$.  We remark that the usual Gram--Schmidt property of $y_{S,T}$ being orthogonal to \emph{all} $z_{T'}$ with $T' \preceq T$ follows by induction in the usual way; this only needs the inductive property that the $z_{T}$'s are orthogonal (not that they're orthonormal).

We may now summarize the discussion so far:
\begin{definition}
    A \emph{positive definiteness problem} occurs at substage $(S,T)$ of modified Gram--Schmidt if either $\pE[y_{S,T}^2] < 0$, or if $T \in \PvZ$ but $\pE[y_{S,\prev(T)} \cdot z_{T}] \neq 0$. (The latter is called a \emph{pseudovariance zero} problem.)
    We say that the modified Gram--Schmidt process \emph{succeeds through substage~$(S,T)$} if it encounters no positive definiteness problem at any substage $(S',T') \preceq_2 (S,T)$.
\end{definition}
\begin{proposition}                                     \label{prop:GS-deal}
    Suppose the modified Gram--Schmidt process succeeds through substage $(S,T)$.  Then we have:
    \begin{itemize}
        \item  $y_{S,T} = {x}^{S} - p(x)$ for some polynomial $p(x)$ supported on monomials ${x}^{T'}$ with $T' \preceq T$;
        \item $\pE[y_{S,T} \cdot z_{T'}] = 0$ for all $T' \preceq T$, and hence $\pE[y_{S,T} \cdot q(x)] = 0$ for all polynomials $q(x)$ supported on monomials $x^{T'}$ with $T' \preceq T$;
        \item $\pE[y_{S,T}^2] \geq 0$.
    \end{itemize}
    In particular, if the process succeeds through stage~$S$, we have:
    \begin{itemize}
        \item $z_S = c \cdot {x}^{S} - p(x)$ for some positive constant $c > 0$ and some polynomial $p(x)$ supported on monomials ${x}^{T}$ with $T \prec S$;
        \item $\spn({x}^{S'} : S' \preceq S) = \spn(z_{S'} : S' \preceq S)$;
        \item $\pE[z_S \cdot z_{T}] = 0$ for all $T \prec S$, and hence $\pE[z_S \cdot q(x)] = 0$ for all polynomials $q(x)$ supported on monomials ${x}^{T'}$ with $T' \prec S$;
        \item $\pE[z_S^2] = 0$ if $S$ is put in $\PvZ$, else $\pE[z_S^2] = 1$.
    \end{itemize}
\end{proposition}

Our main Theorem~\ref{thm:main} follows provided the modified Gram--Schmidt process succeeds through all substages in~$\calM^{\leq D}_2$.  The reason is that then any multilinear  $p(x)$ of degree at most~$D$ can be expressed as $p(x) = \sum_{|T| \leq D} c_T z_T$.  This implies
    \[
        \pE[p(x)^2] = \sum_{|T|, |T'|\leq D} c_T c_{T'} \pE[z_T \cdot z_{T'}] = \sum_{\substack{|T| \leq D \\ T \not \in \PvZsubscript}} c_T^2 \geq 0,
    \]
using Proposition~\ref{prop:GS-deal}.

\subsection{Advanced accounting}

\begin{definition}
    A \emph{\subfactorplus} is defined to be a \subfactor, together with zero or more isolated variable-vertices.
\end{definition}
We still have that the union of \subfactorpluses is a \subfactorplus.  We extend the $\cons(H)$ and $\vbls(H)$ notation to \subfactorpluses, and also the planted distribution notation~$\eta_H$ (being the same as $\eta_{H'}$ where $H'$ is formed from~$H$ by deleting its isolated vertices).
\begin{definition}
    For a \subfactorplus~$H$, we extend the definition of revenue  by assigning \emph{two} credits for all isolated variable-vertices in~$H$.
\end{definition}
\begin{remark}  \label{rem:subplus}
    If $H$ is a \subfactorplus and $H'$ is the \subfactor formed by deleting isolated vertices, then $\cons(H') = \cons(H)$, $C(H') = C(H)$, and $R(H') \leq R(H)$.  Thus the Plausibility Assumption immediately implies that all \subfactorpluses with at most $2 \cdot \CONSMALL$ constraints are also plausible.
\end{remark}

\begin{lemma}       \label{lem:union-revenue0}
   Let $H$ be a \smallish \subfactorplus with $R(H) \leq r$.  Let $H'$ be a \smallish \subfactor with at most~$s$ leaf variables that are not in~$H$.  Assume $r+s \leq \price \cdot \CONSMALL$.  Then $H \cup H'$ is \smallish and satisfies $R(H \cup H') \leq r+s$.
\end{lemma}
\begin{proof}
    Adding $H'$ into $H$ cannot remove any of the debits of~$H$, and the only additional credits that can be created come from the~$s$ leaf variables in $H'$ that are not in~$H$.  (Since $H'$ is only a \subfactor it has no isolated variables.)  This establishes $R(H \cup H') \leq r+s$.  The \smallness conclusion follows immediately from Lemma~\ref{lem:smallness-trick} (here it does not matter that $H$ is a \subfactorplus).
\end{proof}
A key aspect to our main theorem will be that in some cases this revenue bound can be improved:
\begin{lemma}       \label{lem:rev-loss}
   In the setup of Lemma~\ref{lem:union-revenue0}, suppose also that $H'$ has $b$ edges that are ``boundary'' for~$H$, in the sense that each has exactly one endpoint in~$H$.  Then in fact $R(H \cup H') \leq r + s - b$.
\end{lemma}
\begin{proof}
    Let $a$ be an edge in~$H'$ with exactly one endpoint, call it $w$, in~$H$.  We show that the addition of this edge to~$H$ causes a drop of~$1$ in revenue.  If $w$ is a constraint-vertex, then this follows because $w$ already had degree at least $\Tcxty$ in~$H$, so $a$ becomes a new excess edge in~$H$, creating a new debit.  So suppose $w$ is a variable-vertex.  If $w$ had degree at least~$2$ in $H$ then~$a$ is again excess and creates a new debit.  If $w$ had degree~$1$ in~$H$ then the addition of~$a$ changes $w$ from a leaf variable to an interior variable, removing~$1$ credit from~$H$.  Finally, if $w$ was isolated in~$H$ then the addition of~$a$ turns it into a leaf variable, again removing~$1$ credit from~$H$.  Repeating this argument for all $b$ boundary edges completes the proof.
\end{proof}

\ignore{
The following lemma shows that taking the closure of a set of variables~$S$, then adding the resulting variables $\vbls(\cl(S))$ and re-closing, does not increase the closure.
\begin{lemma} \label{lem:iterated-closure-is-closure}
Let $|S| \leq \price \cdot \CONSMALL$.  Then $\cl(S \cup \vbls(\cl(S))) = \cl(S)$.
\end{lemma}
\begin{proof}
By Remark~\ref{rem:ST-closure} it suffices to prove $\cl(S \cup \vbls(\cl(S))) \subseteq \cl(S)$.  By definition of $\cl(S)$, it suffices to prove that $\cl(S \cup \vbls(\cl(S)))$ is $S$-closed and that it is \smallish.  Every leaf vertex of $\cl(S \cup \vbls(\cl(S)))$ must also be a leaf vertex of $\cl(S)$, since $\cl(S) \subseteq \cl(S \cup \vbls(\cl(S)))$ (Remark~\ref{rem:ST-closure}); but every leaf vertex of $\cl(S)$ is in~$S$, and hence $\cl(S \cup \vbls(\cl(S)))$ is indeed $S$-closed.

either in $S$ or in $\vbls(\cl(S))$; but every variable in $\cl(S)$

To see (1), first note that $\cl(S) \subseteq \cl(S \cup \vbls(\cl(S)))$ by Remark~\ref{rem:ST-closure}.  This implies that any leaf vertex of $\cl(S \cup \vbls(\cl(S)))$ must be a leaf vertex of $\cl(S)$. Since $\cl(S)$ is $S$-closed, $\cl(S \cup \vbls(\cl(S)))$ must also be $S$-closed.

For (2), observe that $\cl(S \cup \vbls(\cl(S))$ has at most $|S|$ leaf variables because it is $S$-closed.  We then apply Lemma~\ref{lem:union-revenue0} with $H = \emptyset$ and $H' = \cl(S \cup \vbls(\cl(S)))$.
\end{proof}
}

\subsection{The key lemma}
\newcommand{\Hnew}{H_{\mathrm{new}}}
\newcommand{\Tnew}{T_{\mathrm{new}}}
\newcommand{\Told}{T_{\mathrm{old}}}
\begin{lemma}                                       \label{lem:KEY}
    Let $y = y(x)$ be a polynomial expression of degree~$d$.  Assume $2d \leq \price \cdot \CONSMALL$ and that $\pE[y \cdot p(x)] = 0$ for all polynomials $p$ of degree strictly less than~$d$.  Let $H$ be a \smallish \subfactorplus with $\vbls(H) \supseteq \vblsp(y)$ and $R(H) \leq r$, where we assume $r + d \leq \price \cdot \CONSMALL$. \todo{Probably could relax this r+d to r, which would make the final theorem statement a teeny bit more elegant, but who cares.}Finally, suppose $T$ is a monomial index with $|T| = d$ such that
    \[
        \pE[y \cdot x^T] \neq 0.
    \]
    Then there exists a \smallish \subfactorplus $\Hnew \supseteq H$ with $\vbls(\Hnew) \supseteq \vblsp(y) \cup T$ and $R(\Hnew) \leq r$. (In writing $\vblsp(y) \cup T$, we are using the abuse described in Notation~\ref{not:abuse}.)
\end{lemma}
\begin{proof}
    Let us define
    \[
        \Tnew = T \setminus \vbls(H), \qquad \Told = T \cap \vbls(H), \qquad B = \cl(\vbls(H) \cup T), \qquad \Hnew = H \cup B.
    \]
    First, we show that the \subfactorplus $\Hnew$ is \smallish; \todo{here's where we used the bound on r+d; note we eventually prove R(H') at most r}it follows that the \subfactor $B$ is also \smallish.
    \begin{claim} \label{claim:hnew-small}
    $\Hnew$ is \smallish.
    \end{claim}
    \begin{proof}
    Write $\cl(\vbls(H) \cup T) = H'_1 \cup \cdots \cup H'_{t}$ for \smallish $(\vbls(H) \cup T)$-closed \subfactors $H'_i$.  Let $H'_{< j} \coloneqq H'_1 \cup \cdots \cup H'_{j-1}$, and let $s_j$ denote the number of leaves of $H_j$ that are not in~$H \cup H'_{< j}$.  Then it is easy to see that $\sum_{j = 1}^ t s_j \leq d$.  Now, iteratively apply Lemma~\ref{lem:union-revenue0} to $H \cup H'_1$, $(H \cup H'_1) \cup H'_2$, $((H \cup H'_1) \cup H'_2) \cup H'_3$, \dots to prove the claim.
    \end{proof}

    Next, observe that we have $\vbls(\Hnew) \supseteq \vbls(H) \supseteq \vblsp(y)$;   therefore to prove the lemma, it suffices to show that $\vbls(\Hnew) \supseteq \Tnew$ and that $R(\Hnew) \leq R(H)$.

    For the first of these, given an $(i,c) \in T$ we write $\ol{x}_i = \indet{c}{i} - \pE[\indet{c}{i}]$ and $\ol{x}^T = \prod_{i \in T} \ol{x}_i$.  Observe that $\ol{x}^T - x^T$ is a polynomial of degree strictly less than~$d$; thus $\pE[y \cdot (\ol{x}^T - x^T)] = 0$ and so $\pE[y \cdot \ol{x}^T] \neq 0$.  Now using Theorem~\ref{thm:consistent-distribution} and $B \supseteq \cl(\vbls(H) \cup T) \supseteq \cl(\vblsp(y \cdot \ol{x}^T))$, we conclude
    \begin{equation}            \label{eqn:pEzero}
        \E_{B}[\by \cdot \ol{\bx}^{T}] \neq 0.
    \end{equation}
    In light of this, we claim that every variable $j \in \Tnew$ must appear as a vertex in~$B$ (and hence in $\vbls(\Hnew)$, as needed).  For if $j \not \in \vbls(B)$, then $\ol{\bx}_j$ is independent of all other random variables $\bx_i$ under~$\eta_B$, and so
    \begin{align*}
        \E_{B}[\by \cdot \ol{\bx}^{T}] &= \E_B[\ol{\bx}_j] \cdot \E_{B}[\by \cdot \ol{\bx}^{\Told} \cdot \ol{\bx}^{\Tnew \setminus \{j\}}] \tag{using $j \notin \vbls(H) \supseteq \vblsp(y)$} \\
        &= \pE[\ol{\bx}_j] \cdot \E_{B}[\by \cdot \ol{\bx}^{\Told} \cdot \ol{\bx}^{\Tnew \setminus \{j\}}] = 0\tag{using $B \supseteq \cl(T) \supseteq \cl(\{j\})$ and $\pE[\ol{\bx}_j] = 0$}
    \end{align*}
    in contradiction to~\eqref{eqn:pEzero}.

    It remains to show that $R(\Hnew) \leq R(H)$, which we will do using Lemma~\ref{lem:rev-loss} (with $H' = B$, and $s = |\Tnew|$, recalling that all of $B$'s leaves are in $\vbls(H) \cup T$).  We must show that the number of ``boundary edges'' --- i.e., edges in~$B$ that have exactly one endpoint in~$H$ --- is at least $|\Tnew|$.  Supposing otherwise, the set
    \[
        V = \{\text{variable-vertices~$v \in B$} : \text{$v$ is incident on a boundary edge}\} \cup \Told
    \]
    would satisfy $|V| < |\Tnew| + |\Told| = |T| \leq d$.  We will show that this contradicts~\eqref{eqn:pEzero}.

     \begin{claim}  \label{claim:disconnect}
        The deletion of variable-vertices~$V$ from~$B$ disconnects all variables in~$T$ from all variables in~$\vbls(H)$ \emph{within~$B$}.  (Note that when a variable does not even appear in a subgraph, it is trivially disconnected from all other variables.)
     \end{claim}
     \begin{proof}
        It suffices to show that deleting~$V$ disconnects $\Tnew$ from $\vbls(H)$ within~$B$, as the vertices of~$\Told$ are already in~$V$.   Suppose $j \in \Tnew$ is connected to some variable~$i \in \vbls(H)$ by a path within~$B$.  Since $j \notin \vbls(H)$, there must be some edge in this path that has exactly one endpoint in~$H$. This edge is a boundary edge, and hence the variable-vertex incident on it is in~$V$.  Thus we have indeed established that \emph{every} path within~$B$ from a variable in~$\Tnew$ to a variable in~$\vbls(H)$ must pass through a variable in~$V$.
     \end{proof}

    Recall that the proof is complete once we show that  $|V| < d$ contradicts~\eqref{eqn:pEzero}.  Now
    \begin{align}
        \E_{B}[\by \cdot \ol{\bx}^{T}]
        &= \E_{B}\Bigl[\by \cdot \ol{\bx}^{T} \cdot \sum_{\vec{c} \in \Omega^{V}} \bone[\bx_i = {c}_i\ \forall i \in V]\Bigr]\nonumber\\
        &= \sum_{\vec{c} \in \Omega^{V}} \E_{B}\bigl[\by \cdot \ol{\bx}^{T} \cdot \bone[\bx_i = {c}_i\ \forall i \in V]\bigr]. \label{eqn:every-0}
    \end{align}
    We claim that every summand above equals~$0$.  The reason is that for each summand~$\vec{c}$, either $\bone[\bx_i = {c}_i\ \forall i \in V]$ is always~$0$ under~$\eta_B$ (establishing the claim), or else we may condition on the event, yielding
    \[
        \E_{B}\bigl[\by \cdot \ol{\bx}^{T} \cdot \bone[\bx_i = {c}_i\ \forall i \in V]\bigr] = \Pr_B[\bx_i = {c}_i\ \forall i \in V] \cdot \E_{B}[\by \cdot \ol{\bx}^{T} \mid \bx_i = {c}_i\ \forall i \in V].
    \]
    By Claim~\ref{claim:disconnect} and the definition of the planted distribution~$\eta_B$ (and $\vblsp(y) \subseteq \vbls(H)$), we have that $\by$ and $\ol{\bx}^T$ are conditionally independent under~$\eta_B$, conditioned on all $(\bx_i : i \in V)$.  Therefore
    \[
        \E_{B}[\by \cdot \ol{\bx}^{T} \mid \bx_i = c_i\ \forall i \in V] = \E_{B}[\by\mid \bx_i = c_i\ \forall i \in V] \cdot \E_{B}[\ol{\bx}^{T}\mid \bx_i = c_i\ \forall i \in V].
    \]
    Combining the previous two equations yields
    \[
        \E_{B}\bigl[\by \cdot \ol{\bx}^T \cdot \bone[\bx_i = c_i\ \forall i \in V]\bigr] =
        \E_{B}\bigl[\by \cdot \bone[\bx_i = c_i\ \forall i \in V]\bigr] \cdot  \E_{B}[\ol{\bx}^T \mid \bx_i = c_i\ \forall i \in V].
    \]
    Finally, using $|V| < d$ we will show that the first factor above is~$0$ (thereby establishing the claim that every term in~\eqref{eqn:every-0} is~$0$, in contradiction to~\eqref{eqn:pEzero}).  To see this, we have
    \[
        \E_{B}\bigl[\by \cdot \bone[\bx_i = c_i\ \forall i \in V]\bigr] = \pE\bigl[y \cdot \prod_{i \in V} \indet{c_i}{i}\bigr]
    \]
    because $\cl(\vblsp(y) \cup V) \subseteq \cl(\vbls(H) \cup \vbls(B)) \subseteq B$, where we used Theorem~\ref{thm:iterated-closure-is-closure}.  
    \todo{remark that here's another place we used 2d at most blah-blah?}But this pseudoexpectation is indeed~$0$ by the lemma's assumption, because $\prod_{i \in V} \indet{c_i}{i}$ is a polynomial expression of degree at most~$|V| < d$.
\end{proof}

\subsection{Gram--Schmidt details}

We wish to show that Gram--Schmidt succeeds through substage $(S,T)$ for all $(S,T) \in \calM^{\leq D}_2$.  We will do this by induction along the order~$\preceq_2$.  The key to showing that no positive definiteness problem is encountered at stage~$(S,T)$ will be the existence of a \emph{witness}:
\begin{definition}
    A \emph{witness} for substage $(S,T) \in \calM^{\leq D}_2$ is defined to be a \smallish \subfactorplus $H_{S,T}$ with $\vbls(H_{S,T}) \supseteq \vblsp(y_{S,T})$ and $R(H_{S,T}) \leq 2D$.
\end{definition}
\begin{remark}      \label{rem:OK-start}
    For any substage of the form $(S,\emptyset)$, we may always take as a witness the \subfactorplus consisting of all variables in~$S$ as isolated vertices.
\end{remark}
As the below proposition shows, witnesses are useful for showing that one kind of positive definiteness problem does not occur.  (They will also assist in showing the other kind does not occur.)
\begin{proposition}                                     \label{prop:witness-okay1}
    The existence of a witness $H_{S,T}$ for substage $(S,T)$ implies $\pE[y_{S,T}^2] \geq 0$.
\end{proposition}
\begin{proof}
    By Lemma~\ref{lem:union-revenue0}, we have that $\ol{H} \coloneqq H_{S,T} \cup \cl(\vblsp(y_{S,T}))$ is \smallish. Thus ${\pE[y_{S,T}^2] = \E_{\ol{H}}[\by_{S,T}^2] \geq 0}$, using Theorem~\ref{thm:consistent-distribution}.
\end{proof}

We now come to our main technical theorem:
\begin{theorem}                                     \label{thm:main-technical}
    Let $(S,T) \in \calM^{\leq D}_2$.  Then:
     \begin{enumerate}[label=(\roman*)]
        \item \label{item:witness} Given any witness $H_{S,\emptyset}$ for substage $(S,\emptyset)$, there is a witness $H_{S,T}$ for substage $(S,T)$ satisfying $H_{S,T} \supseteq H_{S,\emptyset}$.
        \item \label{item:success} The Gram--Schmidt process succeeds through substage~$(S,T)$.
     \end{enumerate}
\end{theorem}
\begin{proof}
    The proof will be by (strong) induction on~$(S,T)$ along~$\preceq_2$.  Observe that in proving part~\ref{item:success} of the theorem, by induction we only need to show that no positive definiteness problem occurs at substage~$(S,T)$.  Further, if we can inductively establish part~\ref{item:witness} of the theorem, then Remark~\ref{rem:OK-start} and Proposition~\ref{prop:witness-okay1} imply that $\pE[y_{S,T}^2] \geq 0$. Thus to also establish part~\ref{item:success}, it would only remain to prove that no ``pseudovariance zero problem'' problem occurs.  Also, observe that the pseudovariance zero problem can never occur when $T = \emptyset$. Thus for substages $(S,\emptyset)$, we only need to establish part~\ref{item:witness} of the theorem statement.  But part~\ref{item:witness} is trivial for $(S,\emptyset)$ substages.  Thus all substages of the form $(S,\emptyset)$ are taken care of, including the base case of the induction (namely substage $(\{(i_0,c_0)\}, \emptyset)$, where~$\{(i_0,c_0)\}$ is the first singleton in the order~$\preceq$).

    Thus it remains to establish, for a particular substage $(S,T)$ with $T \neq \emptyset$, that part~\ref{item:witness} of the theorem statement holds, and also that no pseudovariance zero problem occurs.  Given any witness $H_{S,\emptyset}$ for substage $(S,\emptyset)$, by induction we may obtain a witness $H_{S,\prev(T)} \supseteq H_{S,\emptyset}$ for substage $(S,\prev(T))$.  We now distinguish two cases.
    \paragraph{Case 1:} $\pE[y_{S,\prev(T)} \cdot z_T] = 0$.  In this case, $y_{S,T} = y_{S,\prev(T)}$ and therefore $\pE[y_{S,T} \cdot z_T] = 0$.  Thus certainly no pseudovariance zero problem occurs, and also we can establish part~\ref{item:witness} of the theorem statement simply by taking $H_{S,T} = H_{S,\prev(T)}$.  Thus the inductive step is completed in this case.
    \paragraph{Case 2:} $\pE[y_{S,\prev(T)} \cdot z_T] \neq 0$.  This is where the main work in the proof occurs.  First, we will show in this case that $T \in \PvZ$ is impossible, and hence the pseudovariance zero problem cannot have occurred.  We can then complete the induction by finding a witness $H_{S,T} \supseteq H_{S,\prev(T)}$ for substage~$(S,T)$.\\

    First, suppose for contradiction that $T \in \PvZ$.  We have that $y_{S,\prev(T)} = {x}^S - q(x)$ for some $q(x)$ supported on monomials ${x}^{T'}$ with $T' \preceq \prev(T) \prec T$.  By Proposition~\ref{prop:GS-deal} and induction, $z_T$ is orthogonal to all such polynomials.  Thus we deduce
    \begin{equation}    \label{eqn:contra1}
        0 \neq \pE[y_{S,\prev(T)} \cdot z_T]  = \pE[{x}^S \cdot z_T] = \pE[{x}^S \cdot y_{T,\prev(T)}],
    \end{equation}
    the last equality because $T \in \PvZ$ and hence $z_T = y_T = y_{T,\prev(T)}$. By induction (and using Remark~\ref{rem:OK-start}), we have a witness $H_{T,\prev(T)}$ for $y_{T,\prev(T)}$. By Lemma~\ref{lem:union-revenue0} (using $2D + |S| \leq 3D \leq \price \cdot \CONSMALL$) we have that $\overline{H} \coloneqq H_{T,\prev(T)} \cup \cl(S)$ is \smallish.  (In writing $\cl(S)$ we used the abuse from Notation~\ref{not:abuse}.).  Now $\vbls(\overline{H}) \supseteq \vbls(H_{T,\prev(T)}) \cup S \supseteq \vblsp({x}^S \cdot y_{T,\prev(T)})$, so by Theorem~\ref{thm:consistent-distribution} we have
    \[
        \pE[{x}^S \cdot y_{T,\prev(T)}] = \E_{\overline{H}}[{\bx}^S \cdot \by_{T,\prev(T)}], \quad \text{and also} \quad  \E_{\overline{H}}[\by_{T,\prev(T)}^2] = \pE[y_{T,\prev(T)}^2]  = \pE[z_T^2] = 0,
    \]
    the last equality because we're assuming $T \in \PvZ$.  But the second identity above shows that $\by_{T,\prev(T)}^2$ is identically~$0$ under $\eta_{\overline{H}}$, meaning the first expression above must be~$0$.  This contradicts~\eqref{eqn:contra1}.

    Having ruled out the pseudovariance zero problem, we can complete the induction by finding a witness $H_{S,T} \supseteq H_{S,\prev(T)}$ for substage~$(S,T)$.  By Proposition~\ref{prop:GS-deal} we have that $z_T = c \cdot x^{T} - p(x)$ for some constant $c > 0$ and some polynomial~$p(x)$ supported on monomials $x^{T'}$ with $T' \preceq \prev(T)$.  Furthermore, $y_{S,\prev(T)}$ is orthogonal to~$p(x)$ under~$\pE[\cdot]$.  Thus, since we are in Case~2, we may deduce that
    \begin{equation}    \label{eqn:contra2}
        \pE[y_{S,\prev(T)} \cdot x^T] \neq 0.
    \end{equation}
    We may now apply Lemma~\ref{lem:KEY} (with $y = y_{S,\prev(T)}$, $H = H_{S,\prev(T)}$, and $r = 2D$) to obtain a \smallish \subfactorplus $\Hnew \supseteq H_{S,\prev(T)}$ with $\vbls(\Hnew) \supseteq \vblsp(y_{S,\prev(T)}) \cup T$ and $R(\Hnew) \leq 2D$.  This $\Hnew$ is \emph{almost} able to serve as the witness for substage $(S,T)$. The only deficiency is that, although it contains all the variables in $y_{S,\prev(T)}$ and~$x^T$, it doesn't necessarily contain all the variables appearing in~$z_T$ --- as it would need to in order to contain all variables in the new $y_{S,T} = y_{S,\prev(T)} - \pE[y_{S,\prev(T)} \cdot z_T]z_T$.  However, we can fix this by induction; we apply the induction hypothesis to substage~$(T,\prev(T))$, \emph{taking $\Hnew$ as the ``given witness $H_{T,\emptyset}$''}.  This produces a witness --- call it $\Hnew'$ --- for substage $(T,\prev(T))$ that satisfies $\Hnew' \supseteq \Hnew$.  This witness $\Hnew'$ now additionally contains all variables in $y_{T,\prev(T)} = z_T$, and therefore it can now serve as the needed witness for substage~$(S,T)$.
\end{proof}

\section{Wrapping things up by setting parameters} \label{sec:parameters}
To prove our main result on weak refutation, Theorem~\ref{thm:main2}, we simply need to combine Theorems~\ref{thm:random-graph} and Theorem~\ref{thm:main}.  Together these give us a pseudoexpectation defined up to degree
\[
    D = \Omega(\gamma) \cdot \price \cdot \frac{n}{\density^{2/(\Tone-\price)}}, \qquad \text{where } \gamma = \tfrac{\beta^{O(1/\Tone)}}{\maxarity \cdot 2^{O(\maxarity/\Tone)}}.
\]
We need to decide how to best set parameters, which we do under the assumption that $\density \geq 10$.

We start with the special but interesting case when $\Tone$ is thought of very large; specifically, $\Tone \geq \Omega(\log \density)$.  This case arises, e.g., for high-arity $\maxarity$-SAT (where $\Tone = \maxarity-2$) with clause density~$2^{\Theta(\maxarity)}$.  In this case, by choosing
 $\price = \frac12 \Tone$ and  $\beta = e^{-O(\maxarity)}$ for our probability bound, we get $D = n/2^{O(\maxarity/\Tone)}$.  Note that if $\Tone = \Theta(\maxarity)$, as it is in the case of $\maxarity$-SAT, then our SOS degree lower bound is linear in~$n$ with absolutely \emph{no} dependence on $\maxarity = \maxarity(n)$ (all the way up to $\maxarity = \Omega(n)$)!

In the more general regime (e.g., when one thinks of $\maxarity$ as ``constant'' and $\density$ as asymptotically large), a good choice for $\price$ is $\frac{1}{\log \density}$,
 which entails
\[
    D = \Omega(\gamma) \cdot \frac{n}{\density^{2/\Tone} \log \density}.
\]
With this setting, Theorem~\ref{thm:random-graph} tells us that with high probability we get a pseudoexpectation satisfying Corollaries~\ref{cor:pE-sats-stuff},~\ref{cor:pE-sats-constraints}.  Thus we have established the following more precise version of Theorem~\ref{thm:main2}:
\begin{theorem}                                     \label{thm:main2a}
Let $P$ be a $k$-ary Boolean predicate and let $\cmplx(P)$ be the minimum integer $3 \leq \tau \leq k$ for which $P$ fails to support a $\tau$-wise uniform distribution.  Then if $\calI$ is a random instance of $\CSP(P^{\pm})$ with $m = \Delta n$ constraints ($\Delta \geq 10$), then except with probability at most~$\beta$, degree-$D$ SOS fails to (weakly) refute~$\calI$, where
\[
    D = \tfrac{\beta^{O(1/\cmplx(P))}}{k \cdot 2^{O(k/\cmplx(P))}} \cdot \frac{n}{\density^{2/(\cmplx(P) - 2)} \log \density}.
\]
The result also holds if $P$ is a predicate over an alphabet of size $q > 2$ (with an appropriate notion of ``literals''), with no change in parameters.
\end{theorem}

Proving our main result on $\delta$-refutation, Theorem~\ref{thm:main-intro}, requires just a little work.
We now imagine that our instance comes from a random $\CSP(P^{\pm})$  as in Theorem~\ref{thm:main-intro}.  As discussed at the end of Section~\ref{sec:constraint-satisfaction}, given $t$ and taking $\tau = t+1$, we have some $t$-wise uniform distribution $\mu$ on $\{\pm 1\}^k$ which is $\delta$-close to being supported on~$P$, where $\delta = \delta_P(t)$.  We assume that all of the constraint distributions $\mu_f$ are now simply equal to~$\mu$, up to the appropriate negation pattern.  Thus a draw from $\mu_f$ satisfies the constraint at~$f$ except with probability at most~$\delta$.

With the parameter settings chosen earlier,  Theorem~\ref{thm:random-graph} tells us moreover that
\begin{equation} \label{eqn:sublin}
    \#\{ \text{nonempty \subfactors~$H$ with $|\cons(H)| \leq 2 \cdot \CONSMALL$} : I(H) \leq \Tcxty - 1\} \leq \density n^{\frac{1+1/\log \density}{2}} = 2^{\frac{\log n}{2 \log \density}} \cdot \frac{m}{\sqrt{n}}.
\end{equation}
Observe that this bound is \emph{always} $o(m)$, and in the very typical case that $\density \geq n^{\Omega(1)}$, the bound is~$O(\frac{m}{\sqrt{n}})$.  Let us see what this bound means for the pseudodistribution.

Supposing~\eqref{eqn:sublin} holds, let $f$ be any constraint-vertex in~$G$, let $S = N(f)$, and let $H_f$ be the (\smallish) \subfactor induced by the edges between~$f$ and~$S$.  Certainly $\cl(S) \supseteq H_f$, but we may ask whether $\cl(S)$ is strictly bigger than~$H_f$.  Suppose this is the case; i.e., there is some \smallish $S$-closed $H \not \subseteq H_f$.  Then $H' = H_f \cup H$ is a \subfactor satisfying $|\cons(H')| \leq 2 \cdot \CONSMALL$.  Furthermore, the number of leaf variables in~$H'$ must be at least~$1$ (else $H'$ is $\emptyset$-closed and hence empty by Fact~\ref{fact:cl-emptyset}) and strictly less than~$\maxarity$ (else $H \setminus H_f$ will be $\emptyset$-closed and hence empty).  Finally, we claim $R(H') \leq \Tcxty-1$. This is because $R(H_f) = \Tcxty$, the addition of~$H$ cannot add any new credits (since all its leaf variables are already in~$H_f$), and in fact the addition of $H$ must cause a drop of at least one in revenue since~$H$ must have at least one edge not in~$H_f$.  (This argument is similar to Lemma~\ref{lem:rev-loss}.)  We conclude that whenever $\cl(N(f)) \neq H_f$, there must exist a nonempty \subfactor~$H'$ with the following properties: (i)~$|\cons(H')| \leq 2 \cdot \CONSMALL$; (ii)~$I(H') \leq R(H') \leq \Tcxty-1$; (iii)~$H'$ has at least one leaf variable; (iv)~all leaves of $H'$ are adjacent to~$f$.

But~\eqref{eqn:sublin} bounds the number of \subfactors with the first two properties above, and every \subfactor with the latter two properties uniquely determines~$f$.  Thus we conclude:
\[
    \#\{\text{constraints~$f$} : \cl(N(f)) \neq H_f\} \leq 2^{\frac{\log n}{2 \log \density}} \cdot \frac{m}{\sqrt{n}}.
\]
Finally, when $\cl(N(f)) = H_f$, observe that the planted distribution $\eta_{\cl(N(f))}$ is just~$\mu_f$, and hence
\[
    \pE\Bigl[1[\text{$\x$ satisfies~$f$}]\Bigr] = \Pr_{\bx \sim \mu_f}\Bigl[1[\text{$\bx$ satisfies $f$}]\Bigr] \geq 1-\delta.
\]
Combining the last two deductions yields
\[
\pE\Bigl[\text{fraction of constraints satisfied}\Bigr] \geq 1 - \delta - 2^{\frac{\log n}{2 \log \density}} \cdot \frac{1}{\sqrt{n}}.
\]
In summary, we have proven the following more precise version of Theorem~\ref{thm:main-intro}:

\begin{theorem} \label{thm:main-intro2}
Let $P$ be a $k$-ary Boolean predicate and let $1 < t \leq k$.  Let $\calI$ be a random instance of $\CSP(P^{\pm})$ with $m = \Delta n$ constraints.  Then except with probability at most~$\beta$, degree-$D$ SOS fails to
$(\delta_P(t) + \epsilon)$-refute~$\calI$, where
\[
	\epsilon = 2^{\frac{\log n}{2 \log \density}} \cdot \frac{1}{\sqrt{n}}, \qquad
    D = \tfrac{\beta^{O(1/t)}}{k \cdot 2^{O(k/t)}} \cdot \frac{n}{\density^{2/(t-1)} \log \density}.
\]
We remark that $\eps = o(1)$ always, and $\eps = O(\frac{1}{\sqrt{n}})$ whenever $\Delta = n^{\Omega(1)}$.  Finally, the result also holds if $P$ is a predicate over an alphabet of size $q > 2$ (with an appropriate notion of ``literals''), with no change in parameters.
\end{theorem}

\begin{remark}
We should mention that in our $\delta$-refutation result Theorem~\ref{thm:main-intro2}, our pseudoexpectation does \emph{not} satisfy ``solution {value $= 1-\delta_0$}'' as a constraint for any $\delta_0 \leq \delta$; it merely has $\pE[\text{solution value}] \geq 1-\delta$.  Achieving the (stronger) former condition is a direction for future work.  By contrast, for our weak refutation result Theorem~\ref{thm:main2}, the pseudoexpectation \emph{does} satisfy all the constraints and hence also satisfies $\pE[\text{solution value}] = 1$ as a constraint.
\end{remark}

\subsection*{Acknowledgment}
 We would like to thank the Institute for Mathematical Sciences, National University of Singapore in 2016; a visit there was where some of the initial research for this work began.

\bibliographystyle{alpha}
\bibliography{bib/ads,bib/custom,bib/dblp,bib/diss,bib/mr,bib/scholar,bib/witmer}

\newcommand{\etalchar}[1]{$^{#1}$}
\def\cprime{$'$} \def\cprime{$'$}
  \def\ocirc#1{\ifmmode\setbox0=\hbox{$#1$}\dimen0=\ht0 \advance\dimen0
  by1pt\rlap{\hbox to\wd0{\hss\raise\dimen0
  \hbox{\hskip.2em$\mathsfiptscriptstyle\circ$}\hss}}#1\else {\accent"17
  #1}\fi} \def\cprime{$'$} \def\cprime{$'$} \def\cprime{$'$} \def\cprime{$'$}
  \def\cprime{$'$} \def\cprime{$'$} \def\cprime{$'$} \def\cprime{$'$}
  \def\cprime{$'$} \def\polhk#1{\setbox0=\hbox{#1}{\ooalign{\hidewidth
  \lower1.5ex\hbox{`}\hidewidth\crcr\unhbox0}}}
  \def\cfac#1{\ifmmode\setbox7\hbox{$\accent"5E#1$}\else
  \setbox7\hbox{\accent"5E#1}\penalty 10000\relax\fi\raise 1\ht7
  \hbox{\lower1.15ex\hbox to 1\wd7{\hss\accent"13\hss}}\penalty 10000
  \hskip-1\wd7\penalty 10000\box7} \def\cprime{$'$} \def\cprime{$'$}
  \def\cprime{$'$} \def\cprime{$'$} \def\cprime{$'$} \def\cprime{$'$}
  \def\cprime{$'$} \def\cprime{$'$} \def\cprime{$'$} \def\cprime{$'$}
  \def\cprime{$'$} \def\cprime{$'$} \def\cprime{$'$} \def\cprime{$'$}
  \def\cprime{$'$} \def\cprime{$'$} \def\cprime{$'$}
\begin{thebibliography}{DKMPG08}

\bibitem[AAM{\etalchar{+}}11]{AAM+11}
Noga Alon, Sanjeev Arora, Rajsekar Manokaran, Dana Moshkovitz, and Omri
  Weinstein.
\newblock Inapproximability of densest $\kappa$-subgraph from average case
  hardness.
\newblock 2011.

\bibitem[AAT05]{AAT05}
Mikhail Alekhnovich, Sanjeev Arora, and Iannis Tourlakis.
\newblock Towards strong nonapproximability results in the
  {L}ov\'{a}sz-{S}chrijver hierarchy.
\newblock In {\em Proceedings of the 37th Annual ACM Symposium on Theory of
  Computing}, pages 294--303, 2005.

\bibitem[ABR12]{ABR12}
Benny Applebaum, Andrej Bogdanov, and Alon Rosen.
\newblock A dichotomy for local small-bias generators.
\newblock In Ronald Cramer, editor, {\em Theory of Cryptography}, volume 7194
  of {\em Lecture Notes in Computer Science}, pages 600--617. Springer Berlin
  Heidelberg, 2012.

\bibitem[ABW10]{ABW10}
Benny Applebaum, Boaz Barak, and Avi Wigderson.
\newblock Public-key cryptography from different assumptions.
\newblock In {\em Proceedings of the 42nd ACM Symposium on Theory of
  Computing}, pages 171--180, 2010.

\bibitem[AGT12]{AGT12}
Noga Alon, Iftah Gamzu, and Moshe Tennenholtz.
\newblock Optimizing budget allocation among channels and influencers.
\newblock In {\em Proceedings of the 21st {I}nternational {C}onference on
  {W}orld {W}ide {W}eb}, pages 381--388, 2012.

\bibitem[AIK06]{AIK06}
Benny Applebaum, Yuval Ishai, and Eyal Kushilevitz.
\newblock Cryptography in $\text{NC}^0$.
\newblock {\em SIAM Journal on Computing}, 36(4):845--888, 2006.

\bibitem[AL16]{AL16}
Benny Applebaum and Shachar Lovett.
\newblock Algebraic {A}ttacks against {R}andom {L}ocal {F}unctions and {T}heir
  {C}ountermeasures.
\newblock In {\em Proceedings of the 48th Annual ACM Symposium on Theory of
  Computing}, pages 1087--1100, 2016.

\bibitem[Ale03]{Ale03}
M.~Alekhnovich.
\newblock More on average case vs approximation complexity.
\newblock In {\em Proceedings of the 44th IEEE Symposium on Foundations of
  Computer Science}, pages 298--307, 2003.

\bibitem[AM08]{AM08}
Per Austrin and Elchanan Mossel.
\newblock Approximation resistant predicates from pairwise independence.
\newblock In {\em Proceedings of the 23rd IEEE Conference on Computational
  Complexity}, pages 249--258, 2008.

\bibitem[AOW15]{AOW15}
Sarah~R. Allen, Ryan O'Donnell, and David Witmer.
\newblock How to refute a random {CSP}.
\newblock In {\em Proceedings of the 56th Annual IEEE Symposium on Foundations
  of Computer Science}, pages 689--708, 2015.

\bibitem[App13]{App13b}
Benny Applebaum.
\newblock Cryptographic hardness of random local functions--survey.
\newblock In {\em 10th {T}heory of {C}ryptography {C}onference}, 2013.

\bibitem[AR01]{AR01a}
Michael Alekhnovich and Alexander~A. Razborov.
\newblock Lower bounds for polynomial calculus: non-binomial case.
\newblock In {\em Proceedings of the 42nd Annual IEEE Symposium on Foundations
  of Computer Science}, pages 190--199. 2001.

\bibitem[BBaH{\etalchar{+}}12]{BBH+12}
Boaz Barak, Fernando G. S.~L. Brand\~{a}o, Aram~W. Harrow, Jonathan Kelner,
  David Steurer, and Yuan Zhou.
\newblock Hypercontractivity, {S}um-of-{S}quares {P}roofs, and their
  {A}pplications.
\newblock In {\em Proceedings of the 44th Annual ACM Symposium on Theory of
  Computing}, pages 307--326, 2012.

\bibitem[BCG{\etalchar{+}}12]{BCGVZ12}
Aditya Bhaskara, Moses Charikar, Venkatesan Guruswami, Aravindan
  Vijayaraghavan, and Yuan Zhou.
\newblock Polynomial integrality gaps for strong sdp relaxations of densest
  $k$-subgraph.
\newblock In {\em Proceedings of the 23rd ACM-SIAM Symposium on Discrete
  Algorithms}, pages 388--405, 2012.

\bibitem[BCK15]{BCK15}
Boaz Barak, Siu~On Chan, and Pravesh~K. Kothari.
\newblock Sum of squares lower bounds from pairwise independence.
\newblock In {\em Proceedings of the forty-sevent annual ACM symposium on
  Theory of computing}, 2015.

\bibitem[BCMV12]{BCMV12}
Aditya Bhaskara, Moses Charikar, Rajsekar Manokaran, and Aravindan
  Vijayaraghavan.
\newblock On quadratic programming with a ratio objective.
\newblock In {\em Proceedings of the 39th International Colloquium on Automata,
  Languages and Programming}, pages 109--120, 2012.

\bibitem[BGMT12]{BGMT12}
Siavosh Bennabas, Konstantinos Georgiou, Avner Magen, and Madhur Tulsiani.
\newblock {SDP} gaps from pairwise independence.
\newblock {\em Theory of Computing}, 8(12):269--289, 2012.

\bibitem[BJK05]{BJK05}
Andrei Bulatov, Peter Jeavons, and Andrei Krokhin.
\newblock Classifying the complexity of constraints using finite algebras.
\newblock {\em SIAM J. Comput.}, 34(3):720--742, 2005.

\bibitem[BKP04]{BKP04}
Punit Bhargava, Sriram~C. Krishnan, and Rina Panigrahy.
\newblock Efficient multicast on a terabit router.
\newblock In {\em Proceedings of the 12th {A}nnual {IEEE} {S}ymposium on {H}igh
  {P}erformance {I}nterconnects}, pages 61--67, 2004.

\bibitem[BKS13]{BKS13}
Boaz Barak, Guy Kindler, and David Steurer.
\newblock On the optimality of semidefinite relaxations for average-case and
  generalized constraint satisfaction.
\newblock In {\em Innovations in Theoretical Computer Science, {ITCS} '13,
  Berkeley, CA, USA, January 9-12, 2013}, pages 197--214, 2013.

\bibitem[BM16]{BM16}
Boaz Barak and Ankur Moitra.
\newblock Noisy {T}ensor {C}ompletion via the {S}um-of-{S}quares {H}ierarchy.
\newblock In {\em Proceedings of the 29th {A}nnual {C}onference on {L}earning
  {T}heory}, pages 417--445, 2016.

\bibitem[BOGH{\etalchar{+}}03]{BGH+03}
Joshua Buresh-Oppenheim, Nicola Galesi, Shlomo Hoory, Avner Magen, and Toniann
  Pitassi.
\newblock Rank bounds and integrality gaps for cutting planes procedures.
\newblock In {\em Proceedings of the 44th Annual IEEE Symposium on Foundations
  of Computer Science}, pages 318--327, 2003.

\bibitem[BQ09]{BQ09}
Andrej Bogdanov and Youming Qiao.
\newblock On the security of {G}oldreich's one-way function.
\newblock In Irit Dinur, Klaus Jansen, Joseph Naor, and Jos\'{e} Rolim,
  editors, {\em Approximation, Randomization, and Combinatorial Optimization:
  Algorithms and Techniques}, volume 5687 of {\em Lecture Notes in Computer
  Science}, pages 392--405. Springer Berlin Heidelberg, 2009.

\bibitem[Bri08]{Bri08}
Patrick Briest.
\newblock Uniform {B}udgets and the {E}nvy-{F}ree {P}ricing {P}roblem.
\newblock In {\em Proceedings of the 35th International Colloquium on Automata,
  Languages and Programming}, pages 808--819. 2008.

\bibitem[BS]{BS16}
Boaz Barak and David Steurer.
\newblock Proofs, beliefs, and algorithms through the lens of sum-of-squares.
\newblock \url{http://sumofsquares.org/public/index.html}.

\bibitem[BS01]{Ben01}
Eli Ben-Sasson.
\newblock {\em Expansion in {P}roof {C}omplexity}.
\newblock PhD thesis, Hebrew University, 2001.

\bibitem[BS14]{BS14}
Boaz Barak and David Steurer.
\newblock Sum-of-squares proofs and the quest toward optimal algorithms.
\newblock {\em arXiv preprint arXiv:1404.5236}, 2014.

\bibitem[BSB02]{BB02}
Eli Ben-Sasson and Yonatan Bilu.
\newblock A gap in average proof complexity.
\newblock {\em Electronic Colloquium on Computational Complexity {(ECCC)}},
  9(3), 2002.

\bibitem[BSI99]{BI99}
Eli Ben-Sasson and Russell Impagliazzo.
\newblock Random {CNF}'s are hard for the polynomial calculus.
\newblock In {\em Proceedings of the 40th Annual IEEE Symposium on Foundations
  of Computer Science}, pages 415--421, 1999.

\bibitem[BSW01]{BSW01}
Eli Ben-Sasson and Avi Wigderson.
\newblock Short proofs are narrow---resolution made simple.
\newblock {\em J. ACM}, 48(2):149--169, 2001.

\bibitem[CD09]{CD09}
Nadia Creignou and Herv{\'e} Daud{\'e}.
\newblock The {SAT}-{UNSAT} transition for random constraint satisfaction
  problems.
\newblock {\em Discrete Math.}, 309(8):2085--2099, 2009.

\bibitem[CLP02]{CLP02}
A~Crisanti, L~Leuzzi, and G~Parisi.
\newblock The 3-sat problem with large number of clauses in the
  $\infty$-replica symmetry breaking scheme.
\newblock {\em Journal of Physics A: Mathematical and General}, 35(3):481,
  2002.

\bibitem[CMVZ12]{CMVZ12}
Julia Chuzhoy, Yury Makarychev, Aravindan Vijayaraghavan, and Yuan Zhou.
\newblock Approximation algorithms and hardness of the {$k$}-route cut problem.
\newblock In {\em Proceedings of the 23rd Annual ACM-SIAM Symposium on Discrete
  Algorithms}, pages 780--799, 2012.

\bibitem[CS88]{CS88}
Va{\v{s}}ek Chv{\'a}tal and Endre Szemer{\'e}di.
\newblock Many hard examples for resolution.
\newblock {\em J. Assoc. Comput. Mach.}, 35(4):759--768, 1988.

\bibitem[Dan15]{Dan15}
Amit Daniely.
\newblock Complexity {T}heoretic {L}imitations on {L}earning {H}alfspaces.
\newblock {\em CoRR}, abs/1505.05800, 2015.

\bibitem[DFHS06]{DFHS06}
Erik~D. Demaine, Uriel Feige, Mohammad~Taghi Hajiaghayi, and Mohammad~R.
  Salavatipour.
\newblock Combination can be hard: {A}pproximability of the unique coverage
  problem.
\newblock In {\em Proceedings of the 17th Annual ACM-SIAM Symposium on Discrete
  Algorithms}, pages 162--171, 2006.

\bibitem[DKMPG08]{DKM08}
Josep Diaz, Lefteris Kirousis, Dieter Mitsche, and Xavier Perez-Gimenez.
\newblock A new upper bound for 3-{SAT}.
\newblock In {\em IARCS Annual Conference on Foundations of Software Technology
  and Theoretical Computer Science}, volume~2, pages 163--174, 2008.

\bibitem[DLSS13]{DLS13}
Amit Daniely, Nati Linial, and Shai Shalev-Shwartz.
\newblock More data speeds up training time in learning halfspaces over sparse
  vectors.
\newblock In {\em Advances in Neural Information Processing Systems}, pages
  145--153, 2013.

\bibitem[DLSS14]{DLS14}
Amit Daniely, Nati Linial, and Shai Shalev-Shwartz.
\newblock From average case complexity to improper learning complexity.
\newblock In {\em Proceedings of the 46th Annual ACM Symposium on Theory of
  Computing}, pages 441--448. ACM, 2014.

\bibitem[DS14]{DSS14}
Amit Daniely and Shai Shalev{-}Shwartz.
\newblock Complexity theoretic limitations on learning {DNF}'s.
\newblock Technical Report 1404.3378, arXiv, 2014.

\bibitem[DSS15]{DSS15}
Jian Ding, Allan Sly, and Nike Sun.
\newblock Proof of the satisfiability conjecture for large $k$.
\newblock In {\em Proceedings of the 47th Annual ACM Symposium on Theory of
  Computing}, pages 59--68, 2015.

\bibitem[Fei02]{Fei02}
Uriel Feige.
\newblock Relations {B}etween {A}verage {C}ase {C}omplexity and {A}pproximation
  {C}omplexity.
\newblock In {\em Proceedings of the 34th Annual ACM Symposium on Theory of
  Computing}, pages 534--543, 2002.

\bibitem[FG01]{FG01}
Joel Friedman and Andreas Goerdt.
\newblock Recognizing more unsatisfiable random 3-{SAT} instances efficiently.
\newblock In {\em Automata, languages and programming}, volume 2076 of {\em
  Lecture Notes in Comput. Sci.}, pages 310--321. Springer, Berlin, 2001.

\bibitem[FKO06]{FKO06}
Uriel Feige, Jeong~Han Kim, and Eran Ofek.
\newblock Witnesses for non-satisfiability of dense random 3{CNF} formulas.
\newblock In {\em Proceedings of the 47th Annual IEEE Symposium on Foundations
  of Computer Science}, pages 497--508, 2006.

\bibitem[FPV15]{FPV15}
Vitaly Feldman, Will Perkins, and Santosh Vempala.
\newblock {O}n the {C}omplexity of {R}andom {S}atisfiability {P}roblems with
  {P}lanted {S}olutions.
\newblock In {\em Proceedings of the 47th Annual ACM Symposium on Theory of
  Computing}, pages 77--86, 2015.

\bibitem[Gab16]{Gab16}
Oliver Gableske.
\newblock {d}imetheus.
\newblock In {\em Proceedings of {SAT} {C}ompetition 2016: {S}olver and
  {B}enchmark {D}escriptions}, pages 37--38, 2016.

\bibitem[GK01]{GK01}
Andreas Goerdt and Michael Krivelevich.
\newblock Efficient recognition of random unsatisfiable {$k$}-{SAT} instances
  by spectral methods.
\newblock In {\em S{TACS} 2001 ({D}resden)}, volume 2010 of {\em Lecture Notes
  in Comput. Sci.}, pages 294--304. Springer, Berlin, 2001.

\bibitem[GL04]{GL04}
Andreas Goerdt and Andr{\'{e}} Lanka.
\newblock An approximation hardness result for bipartite {C}lique.
\newblock {\em Electronic Colloquium on Computational Complexity ({ECCC})},
  (048), 2004.

\bibitem[Gol00]{Gol00}
Oded Goldreich.
\newblock Candidate {O}ne-{W}ay {F}unctions {B}ased on {E}xpander {G}raphs.
\newblock In {\em Electronic Colloquium on Computational Complexity (ECCC)},
  volume~7, 2000.

\bibitem[Gri01]{Gri01}
Dima Grigoriev.
\newblock Complexity of positivstellensatz proofs for the knapsack.
\newblock {\em Computational Complexity}, 10(2):139--154, 2001.

\bibitem[Hua13]{Hua13}
Sangxia Huang.
\newblock Approximation resistance on satisfiable instances for predicates with
  few accepting inputs (extended abstract).
\newblock In {\em Proceedings of the 45th Annual ACM Symposium on Theory of
  Computing}, pages 457--466, 2013.

\bibitem[Hua14]{Hua14}
Sangxia Huang.
\newblock Approximation {R}esistance on {S}atisfiable {I}nstances for
  {P}redicates with {F}ew {A}ccepting {I}nputs.
\newblock {\em Theory of {C}omputing}, 10(14):359--388, 2014.

\bibitem[IKOS08]{IKOS08}
Yuval Ishai, Eyal Kushilevitz, Rafail Ostrovsky, and Amit Sahai.
\newblock Cryptography with constant computational overhead.
\newblock In {\em Proceedings of the 40th ACM Symposium on Theory of
  Computing}, pages 433--442, 2008.

\bibitem[KM16]{KM16}
Subhash Khot and Dana Moshkovitz.
\newblock Candidate hard unique game.
\newblock In {\em Proceedings of the 48th Annual ACM Symposium on Theory of
  Computing}, pages 63--76, 2016.

\bibitem[KOTZ14]{KOTZ14}
Manuel Kauers, Ryan O'Donnell, Li-Yang Tan, and Yuan Zhou.
\newblock Hypercontractive inequalities via {SOS}, and the {F}rankl-{R}\"odl
  graph.
\newblock In {\em Proceedings of the 25th Annual ACM-SIAM Symposium on Discrete
  Algorithms}, pages 1644--1658, 2014.

\bibitem[Lau09]{Lau09}
Monique Laurent.
\newblock Sums of squares, moment matrices and optimization over polynomials.
\newblock In {\em Emerging applications of algebraic geometry}, volume 149 of
  {\em IMA Vol. Math. Appl.}, pages 157--270. Springer, New York, 2009.

\bibitem[LRS15]{LRS15}
James~R Lee, Prasad Raghavendra, and David Steurer.
\newblock Lower bounds on the size of semidefinite programming relaxations.
\newblock In {\em Proceedings of the forty-seventh annual ACM symposium on
  Theory of computing}. ACM, 2015.

\bibitem[MPRT16]{MPR16}
Raffaele Marino, Giorgio Parisi, and Federico Ricci-Tersenghi.
\newblock The backtracking survey propagation algorithm for solving random
  {K}-{SAT} problems.
\newblock {\em Nature Communications}, 7(12996), 2016.

\bibitem[MST03]{MST03}
Elchanan Mossel, Amir Shpilka, and Luca Trevisan.
\newblock On $\epsilon$-biased generators in $\text{NC}^0$.
\newblock In {\em Proceedings of the 44th IEEE Symposium on Foundations of
  Computer Science}, pages 136--145, 2003.

\bibitem[MW16]{MW16}
Ryuhei Mori and David Witmer.
\newblock Lower bounds for {CSP} refutation by {SDP} hierarchies.
\newblock In {\em RANDOM '16}, 2016.

\bibitem[OW14]{OW14}
Ryan O'Donnell and David Witmer.
\newblock Goldreich's {PRG}: Evidence for near-optimal polynomial stretch.
\newblock In {\em Proceedings of the 29th Annual Conference on Computational
  Complexity}, pages 1--12, 2014.

\bibitem[OWWZ14]{OWWZ14}
Ryan O'Donnell, John Wright, Chenggang Wu, and Yuan Zhou.
\newblock Hardness of robust graph isomorphism, {L}asserre gaps, and asymmetry
  of random graphs.
\newblock In {\em Proceedings of the 25th Annual ACM-SIAM Symposium on Discrete
  Algorithms}, pages 1659--1677, 2014.

\bibitem[OZ13]{OZ13}
Ryan O'Donnell and Yuan Zhou.
\newblock Approximability and proof complexity.
\newblock In {\em Proceedings of the Twenty-Fourth Annual ACM-SIAM Symposium on
  Discrete Algorithms}, pages 1537--1556. SIAM, 2013.

\bibitem[Rag08]{Rag08}
Prasad Raghavendra.
\newblock Optimal {A}lgorithms and {I}napproximability {R}esults for {E}very
  {CSP}?
\newblock In {\em Proceedings of the 40th Annual ACM Symposium on Theory of
  Computing}, pages 245--254, 2008.

\bibitem[RRS16]{RRS16}
Prasad Raghavendra, Satish Rao, and Tselil Schramm.
\newblock Strongly refuting random csps below the spectral threshold.
\newblock {\em CoRR}, abs/1605.00058, 2016.

\bibitem[RSW16]{RSW16}
Ilya Razenshteyn, Zhao Song, and David~P. Woodruff.
\newblock Weighted low rank approximations with provable guarantees.
\newblock In {\em Proceedings of the 48th Annual ACM Symposium on Theory of
  Computing}, pages 250--263, 2016.

\bibitem[SAT]{SAT14}
\url{http://satcompetition.org/2014/certunsat.shtml}.

\bibitem[Sch08]{Sch08}
Grant Schoenebeck.
\newblock {L}inear {L}evel {L}asserre {L}ower {B}ounds for {C}ertain
  $k$-{CSP}s.
\newblock In {\em Proceedings of the 49th Annual IEEE Symposium on Foundations
  of Computer Science}, pages 593--602, 2008.

\bibitem[Tul09]{Tul09}
Madhur Tulsiani.
\newblock {CSP} gaps and reductions in the lasserre hierarchy.
\newblock In {\em Proceedings of the 41st Annual {ACM} Symposium on Theory of
  Computing, {STOC} 2009, Bethesda, MD, USA, May 31 - June 2, 2009}, pages
  303--312, 2009.

\bibitem[TW13]{TW13}
Madhur Tulsiani and Pratik Worah.
\newblock ${LS}_+$ lower bounds from pairwise independence.
\newblock In {\em Proceedings of the 28th Annual Conference on Computational
  Complexity}, pages 121--132, 2013.

\end{thebibliography}

\appendix

\section{Proof that random  graphs satisfy the Plausibility Assumption} \label{app:expansion}

Here we prove Theorem~\ref{thm:random-graph}, which we restate for convenience:
\paragraph{Theorem~\ref{thm:random-graph} restated.}\emph{Let $\Tone = \Tcxty - 2 \geq 1$.   Fix $0 < \price \leq .99 \Tone$,  $0 < \beta < \frac12$. Then except with probability at most $\beta$, when $\bG$~is a random instance with $m = \density n$ constraints, the Plausibility Assumption holds provided
    \begin{equation}    \label{eqn:small-value}
        \CONSMALL \leq \gamma \cdot \frac{n}{\density^{2/(\Tone-\price)}},
    \end{equation}
    where $\gamma = \frac{1}{\maxarity} \left(\frac{\beta^{1/\Tone}}{2^{\maxarity/\Tone}}\right)^{O(1)}$.   Moreover, assuming $\price < 1$, except with probability at most~$\beta$ we have
    \begin{equation}    \label{eqn:dont-fail-me}
        \#\{\textnormal{nonempty \subfactors~$H$ with $\cons(H) \leq 2 \cdot \CONSMALL$} : I(H) \leq \Tcxty - 1\} \leq \density n^{\frac{1+\price}{2}}.
    \end{equation}
}

\begin{proof}
    A remark before we begin: the expression in~\eqref{eqn:small-value} was chosen precisely so that
    \begin{equation} \label{eqn:weirdness}
        c \leq 2 \cdot \CONSMALL \implies 20^\maxarity \cdot \density \cdot \left(\tfrac{\maxarity c}{n}\right)^{\frac{\Tone-\price}{2}} \leq \beta/50^\maxarity,
    \end{equation}
    provided the $O(1)$ in the definition of~$\gamma$ is a sufficiently large universal constant.

    The proof is a standard argument of the kind used to show that a random bipartite graph has good expansion.   Fixing $I_0 \in \{0, \Tcxty - 1\}$, $1 \leq c \leq 2 \cdot \CONSMALL$,  and~$1 \leq v \leq \maxarity c$, let us upper-bound
    \begin{equation}    \label{eqn:cv-bound1}
        \E[\#\{ \text{\subfactors with $c$ constraints, $v$ vertices, and income at most~$I_0$}\}].
    \end{equation}
    There are $\binom{m}{c}$ choices for the constraints and $\binom{n}{v}$ choices for the variables.  Then by using Lemma~\ref{lem:alternate-accounting},
    \begin{equation}    \label{eqn:cv-bound2}
        \eqref{eqn:cv-bound1} \leq \binom{m}{c}\binom{n}{v} \Pr[\text{fixed set of~$c$ constraints and~$v$ variables gets at least~$A$ edges}],
    \end{equation}
    where
    $
        A \coloneqq \frac{\Tcxty - \price}{2} \cdot c + v - \frac{I_0}{2} 
    $.
    In~\eqref{eqn:cv-bound2}, we may imagine that a constraint's variables are chosen uniformly and independently (i.e., \emph{without} conditioning on them being distinct), as this only increases the probability in question.  Now any fixed set of~$c$ constraints has at most~$\maxarity c$ edges coming out it, so the probability that some integer $a > A$ of them will go into a fixed set of~$v$ variables is at most
    \[
        \binom{\maxarity c}{a} \cdot \left(\frac{v}{n}\right)^a \leq 2^{\maxarity c} \cdot \left(\frac{v}{n}\right)^a  \leq 2^{\maxarity c}  \cdot \left(\frac{v}{n}\right)^A.
    \]
    Thus
    \begin{align}
        \eqref{eqn:cv-bound2} \leq 2^{ \maxarity c} \binom{m}{c}\binom{n}{v}\left(\frac{v}{n}\right)^A
        \leq 2^{\maxarity c} \left(\frac{em}{c}\right)^c\left(\frac{en}{v}\right)^v\left(\frac{v}{n}\right)^A
        &=  \left(e2^\maxarity e^{v/c} (v/c)\right)^c \cdot \density^c \cdot \left(\frac{v}{n}\right)^{\frac{\Tone-\price}{2}\cdot c - I_0/2} \nonumber\\
        &\leq \left(20^\maxarity\right)^c \cdot \density^c \cdot  \left(\tfrac{\maxarity c}{n}\right)^{\frac{\Tone-\price}{2}\cdot c - I_0/2}, \label{eqn:the-bound}
    \end{align}
    where the equality used the definition of~$A$ and the subsequent inequality used $v \leq \maxarity c$.

    We now split into two cases, depending on whether $I_0$ is~$0$ or $\Tcxty-1$.  When $I_0 = 0$ we use
    \[
        \eqref{eqn:cv-bound1} \leq \eqref{eqn:the-bound} = \Bigl(20^\maxarity \cdot \density \cdot \left(\tfrac{\maxarity c}{n}\right)^{\frac{\Tone - \price}{2}}\Bigr)^c \leq \left(\tfrac{\beta}{50^\maxarity}\right)^c,
    \]
    using~\eqref{eqn:weirdness}. Summing over the at most $\maxarity c$ possibilities for~$v$ gives
    \[
        \E[\#\{ \text{\subfactors with $c$ constraints and income at most~$0$}\}] \leq \maxarity c \left(\tfrac{\beta}{50^\maxarity}\right)^c.
    \]
    Now summing this expression over all $1 \leq c \leq 2 \cdot \CONSMALL$ we get
    \[
         \E[\#\{\textnormal{implausible \subfactors~$H$} : |\cons(H)| \leq 2 \cdot \CONSMALL\}]
         \leq \sum_{c=1}^\infty \maxarity c \left(\tfrac{\beta}{50^\maxarity}\right)^c \leq \beta.
    \]
    Thus Markov's inequality  implies that the Plausibility Assumption holds except with probability at most~$\beta$.

    The analysis for $I_0 = \Tcxty-1$ is similar.  In this case, we use
    \[
        \eqref{eqn:cv-bound1} \leq \eqref{eqn:the-bound} = \Bigl(20^\maxarity \cdot \density \cdot \left(\tfrac{\maxarity c}{n}\right)^{\frac{\Tone - \price}{2}}\Bigr)^{c-1} \cdot 20^\maxarity \cdot \density \cdot \left(\tfrac{n}{\maxarity c}\right)^{\frac{1+\price}{2}}  \leq \left(\tfrac{\beta}{50^\maxarity}\right)^{c-1} \cdot 20^\maxarity \cdot \density n^{\frac{1+\price}{2}},
    \]
    again using~\eqref{eqn:weirdness}.  We again sum this over the at most $\maxarity c$ possibilities for~$v$.  We also only need to sum this over all $c \geq 2$, since if $\cons(H) = 1$ then $I(H) = \Tcxty - \price > \Tcxty - 1$.  We then obtain
    \begin{multline*}
        \E[\#\{\textnormal{nonempty \smallish \subfactors~$H$ with $|\cons(H)| \leq 2 \cdot \CONSMALL$} : I(H) \leq \Tcxty - 1\}]  \\ \leq \sum_{c=2}^\infty \maxarity c \left(\tfrac{\beta}{50^\maxarity}\right)^{c-1} 20^\maxarity \cdot \density n^{\frac{1+\price}{2}} \leq \beta \cdot n^{\frac{1+\price}{2}},
    \end{multline*}
    and again Markov's inequality establishes that~\eqref{eqn:dont-fail-me} holds except with probability at most~$\beta$.
\end{proof}

\end{document}